\documentclass[namedate,webpdf,imamci]{ima-authoring-template-arxiv}%
\usepackage[T1]{fontenc}
\DeclareFontShape{T1}{ptm}{m}{scit}{<->ssub * ptm/m/sc}{}

\graphicspath{{Fig/}}

\theoremstyle{thmstyletwo}%
\newtheorem{theorem}{Theorem}[section]
\newtheorem{proposition}{Proposition}[section]

\newtheorem{remark}[theorem]{Remark}
\newtheorem{corollary}[theorem]{Corollary}              
\newtheorem{assumption}[theorem]{Assumption}
\newtheorem{definition}[theorem]{Definition}

\numberwithin{equation}{section}

\renewcommand{\qed}{\hfill\ensuremath{\opensquare}}
\usepackage{abbreviations}
\usepackage{cite}
\usepackage{subfig}
\usepackage{pgfplots}
\pgfplotsset{compat=newest}

\begin{document}

\title[Koopman meets input-output data]{Koopman meets input-output data: Data-driven output-feedback control of nonlinear systems with closed-loop guarantees}

\author{Robin Strässer\ORCID{0000-0003-4629-844X}${}^{*,a}$, Julian Berberich\ORCID{0000-0001-6366-6238}${}^{a}$, Manuel Schaller\ORCID{0000-0002-8081-5108}${}^{b}$, Karl Worthmann\ORCID{0000-0002-1450-2373}${}^{c}$, and Frank Allgöwer\ORCID{0000-0002-3702-3658}${}^{a}$
    \vspace*{0.5\baselineskip}
    \address{\orgdiv{${}^{a}$University of Stuttgart}, \orgname{Institute for Systems Theory and Automatic Control}, \orgaddress{\street{Pfaffenwaldring 9}, \postcode{70550 Stuttgart}, \state{Baden-Württemberg}, \country{Germany}}}
    \address{\orgdiv{${}^{b}$Chemnitz University of Technology}, \orgname{Faculty of Mathematics},\\\orgaddress{\street{Reichenhainer Str. 41}, \postcode{09126 Chemnitz}, \state{Sachsen}, \country{Germany}}}
    \address{\orgdiv{${}^{c}$Technische Universität Ilmenau}, \orgname{Institute of Mathematics, Optimization-based Control Group}, \orgaddress{\street{Weimarer Str. 25}, \postcode{99693 Ilmenau}, \state{Thüringen}, \country{Germany}}\vspace*{0.5\baselineskip}}
}

\authormark{Strässer et al. (2026)}

\corresp[*]{Corresponding author: \href{email:robin.straesser@ist.uni-stuttgart.de}{robin.straesser@ist.uni-stuttgart.de}}

\received{Date}{0}{Year}
\revised{Date}{0}{Year}
\accepted{Date}{0}{Year}

\abstract{%
    Data-driven control of nonlinear systems from input-output measurements remains a fundamental challenge, as existing approaches with rigorous closed-loop guarantees predominantly require access to full state measurements. 
    In this paper, we address this gap by proposing a data-driven output-feedback controller design method for nonlinear systems that provides provable closed-loop guarantees while operating solely on measured input-output data. 
    Our approach combines Koopman operator theory with an extended state representation of the nonlinear system constructed from input-output trajectories. 
    This allows us to obtain a bilinear surrogate model directly from data, on which robust state-feedback design methods can be applied. 
    By exploiting the observability of the underlying nonlinear system, we establish exponential stability of the extended state, which in turn implies exponential convergence of the original system state to the origin. 
    Finally, we validate our theoretical findings in numerical simulations.%
}
\keywords{Koopman operator; Input-output data; Nonlinear data-driven control; Subspace identification.}

\maketitle

\section{Introduction}
Data-driven control of dynamical systems has emerged as a powerful paradigm for designing controllers directly from measured data, bypassing the need for first-principles model derivation.
This is particularly appealing in practice, where complex system dynamics are often difficult or expensive to model analytically, yet plenty of measurement data are readily available.
For linear systems, data-driven control is by now relatively well understood.
Willems' fundamental lemma~\citep{willems:rapisarda:markovsky:demoor:2005} provides a non-parametric characterization of all trajectories of a linear time-invariant system in terms of a single persistently excited experiment, forming the basis of a large body of work on data-driven predictive control and stabilization~\citep{markovsky:rapisarda:2008,coulson:lygeros:dorfler:2019,berberich:kohler:muller:allgower:2020c,faulwasser:ou:pan:schmitz:worthmann:2023}.
Extensions to output-feedback settings for linear systems have been pursued in several directions, including robust output-feedback controllers~\citep{berberich:scherer:allgower:2023}, verification of dissipativity properties from input-output data~\citep{koch:berberich:allgower:2022}, and data-driven output-feedback control of linear MIMO systems~\citep{alsalti:lopez:muller:2025}. 
However, extending these results to the nonlinear setting remains substantially more challenging and is largely open.

For nonlinear systems, data-driven control with closed-loop guarantees typically relies on state or state-derivative data~\citep{martin:schon:allgower:2023b}, and extensions to the output-feedback case are rare.
A notable exception is the work of~\citet{dai:depersis:monshizadeh:tesi:2023,dai:depersis:monshizadeh:tesi:2025}, who design dynamic output feedback controllers for discrete-time nonlinear systems directly from input-output data with local stability guarantees, requiring that an auxiliary input-output representation of the system can be expressed through a known dictionary of basis functions.
Constructing such representations from input-output data in a principled way remains an open challenge, and a natural starting point is the identification of nonlinear systems from data, which has been studied extensively~\citep{noel:kerschen:2017,ljung:andersson:tiels:schon:2020}.
Classical approaches include lifting to higher-dimensional feature spaces via basis function expansions, such as LPV~\citep{verdult:verhaegen:2002,bamieh:giarre:2002} or polynomial approximations~\citep{paduart:lauwers:swevers:smolders:schoukens:pintelon:2010}, local or piecewise linear subspace identification~\citep{verdult:2002}, and probabilistic approaches based on expectation-maximization~\citep{schon:wills:ninness:2011} or Gaussian processes~\citep{frigola:lindsten:schon:rasmussen:2013,frigola:chen:rasmussen:2014}.
A particularly relevant subclass is the class of bilinear systems, which arise naturally as finite-dimensional Koopman representations of nonlinear dynamics~\citep{iacob:toth:schoukens:2024} and for which subspace identification methods have been developed~\citep{favoreel:demoor:vanoverschee:1999,verdult:verhaegen:1999,chen:maciejowski:1999,chen:maciejowski:2000,favoreel:1999}, including kernel-based variants that avoid the exponential growth of data matrices with system order~\citep{verdult:verhaegen:2005,wingerden:verhaegen:2009}.
Finite-sample guarantees for bilinear system identification have also been studied recently~\citep{sattar:oymak:ozay:2022,sattar:jedra:fazel:dean:2025}, direct data-driven controller design has been proposed via LMIs~\citep{bisoffi:depersis:tesi:2020a} and model predictive control~\citep[MPC;][]{xie:berberich:strasser:allgower:2025}, and first end-to-end guarantees from identification to closed-loop control have been established in~\citet{chatzikiriakos:strasser:allgower:iannelli:2026}.

Complementary to system identification, observer design for nonlinear and bilinear systems is a classical and active research area; see~\citet{besancon:2007a,besancon:2007b,isidori:2017,bernard:2019,bernard:andrieu:astolfi:2022}.
Fundamental notions of nonlinear observability have been established from geometric~\citep{nijmeijer:1982,nijmeijer:vanderschaft:1990,hermann:krener:1977} and system-theoretic~\citep{sontag:1998,isidori:1985} perspectives, with further results for polynomial~\citep{sontag:1979} and discrete-time systems~\citep{albertini:dallessandro:1996}.
For nonlinear systems more broadly, early results on state reconstruction and convergence of the estimation error are due to~\citet{thau:1973}.
For bilinear systems specifically, observability conditions~\citep{williamson:1977} and stable state estimators have been developed~\citep{funahashi:1979,hara:furuta:1976,bornard:couenne:celle:1989,elliott:2009}, though observer convergence in this setting typically depends on the applied input sequence.

A promising framework for handling nonlinear dynamics in a data-driven fashion is Koopman operator theory~\citep{koopman:1931,mezic:2005,mauroy:mezic:susuki:2020,bevanda:sosnowski:hirche:2021,brunton:budisic:kaiser:kutz:2022}, which lifts nonlinear system dynamics into an infinite-dimensional, linear space through the action of the Koopman operator on observable functions.
This viewpoint has sparked a rich literature on data-driven approximation of the Koopman operator, including extended dynamic mode decomposition~\citep[EDMD;][]{williams:kevrekidis:rowley:2015,schmid:2010}, kernel-based variants~\citep{williams:rowley:kevrekidis:2016,klus:nuske:hamzi:2020,philipp:schaller:worthmann:peitz:nuske:2024}, and deep learning approaches~\citep{lusch:kutz:brunton:2018,yeung:kundu:hodas:2019,otto:rowley:2019}. 
Crucially, it has been established that linear Koopman approximations are generally insufficient for controlled systems, and that bilinear structures naturally emerge in Koopman representations of input-affine nonlinear systems~\citep{iacob:toth:schoukens:2024}; see also~\citet{haseli:cortes:2023a,haseli:cortes:2023b,haseli:cortes:2023c,haseli:cortes:2026,shang:haseli:cortes:zheng:2026} for further theoretical developments on the structure of controlled Koopman operators.
Further,~\citet{lian:wang:jones:2021,shang:cortes:zheng:2024,xiong:yuan:miao:wang:cortes:papachristodoulou:2025} propose extensions of Willems' fundamental lemma to the nonlinear setting using Koopman embeddings, though typically under restrictive invariance assumptions on the chosen dictionary.
Recently,~\citet{lazar:2025} addresses these limitations by constructing the Koopman operator on a product Hilbert space formed as the tensor product of state and input observable spaces, relaxing dictionary invariance and measure preservation requirements.
This allows the author to derive a nonlinear fundamental lemma for input-\emph{output} data by combining the resulting exact \emph{infinite}-dimensional bilinear representation with Hankel operators and a frame-based persistency of excitation condition.
However, the approach inherently relies on infinite-dimensional representations and requires state measurements during the offline data collection phase to construct the lifted Hankel operator.
While finite-dimensional EDMD approximations via a Khatri-Rao scheme are proposed, no finite-sample error bounds or closed-loop guarantees are established, limiting its applicability to rigorous data-driven controller design.

A key insight for practical applicability is that the Koopman operator can be approximated using delay-coordinate embeddings~\citep{mezic:banaszuk:2004,robinson:2005,susuki:mezic:2015,arbabi:mezic:2017b,mezic:2022,koltai:kunde:2024}, which are directly constructed from input-output time-series data without requiring state measurements.
These findings directly led to Hankel DMD~\citep{arbabi:mezic:2017b,kamb:kaiser:brunton:kutz:2020,pan:duraisamy:2020} and latent EDMD~\citep{ouala:chapron:collard:gaultier:fablet:2023}.
Koopman operator approximations based on delay-coordinate embeddings connect naturally to Takens' theorem~\citep{takens:1981} and have been exploited in MPC~\citep{korda:mezic:2018a} and various engineering applications such as flow prediction~\citep{yuan:zhou:zhou:wen:liu:2021}, soft robotics~\citep{bruder:fu:gillespie:remy:vasudevan:2021,haggerty:banks:kamenar:cao:curtis:mazic:hawkes:2023}, and grip force prediction~\citep{bazina:kamenar:fonoberova:mezic:2025}.
Finite-data error bounds for Koopman approximations are provided in~\citet{mezic:2022,nuske:peitz:philipp:schaller:worthmann:2023,schaller:worthmann:philipp:peitz:nuske:2023,zhang:zuazua:2023,yadav:mauroy:2025} for EDMD and in~\citet{philipp:schaller:worthmann:peitz:nuske:2024,kurdila:paruchuri:powell:guo:bobade:estes:wang:2024,kohne:philipp:schaller:schiela:worthmann:2025,bold:philipp:schaller:worthmann:2025,philipp2025error,strasser:schaller:berberich:worthmann:allgower:2025} for kernel-based EDMD variants.
Moreover,~\citet{strasser:schaller:worthmann:berberich:allgower:2025,strasser:schaller:worthmann:berberich:allgower:2026,strasser:berberich:allgower:2025,strasser:berberich:schaller:worthmann:allgower:2025} establish first closed-loop guarantees for Koopman-based controllers, including Koopman-based MPC~\citep{worthmann:strasser:schaller:berberich:allgower:2024,bold:grune:schaller:worthmann:2025,bold:schaller:schimperna:worthmann:2025,schimperna:bold:kohler:worthmann:2026,schimperna2025data}; see the recent overview paper~\citet{strasser:worthmann:mezic:berberich:schaller:allgower:2026}.
However, all of the above results with closed-loop guarantees require state measurements for the construction of the Koopman surrogate model and the subsequent controller design.
Obtaining full state measurements is often impractical or infeasible in real-world applications, where only input-output data are available, motivating the need for output-feedback approaches that avoid this restrictive assumption.

On the observer and estimation side, the Koopman framework has been used to design state estimators via various lifting strategies, including bilinear~\citep{surana:2016,surana:banaszuk:2016,surana:2020,otto:peitz:rowley:2024}, linear observable~\citep{lei:yin:2025,lyu:lang:wang:2025,yang:gao:chen:lv:wang:2025}, and dual Koopman forms~\citep{mohet:mauroy:winkin:2025}, as well as kernel-based and neural-network-enhanced Kalman filters~\citep{netto:mili:2018a,netto:mili:2018b,jiang:zhang:zuo:shi:su:2022,huang:zheng:fettweis:2024,guo:korotkine:forbes:barfoot:2021}; see~\citet{otto:rowley:2021,shi:haseli:mamakoukas:bruder:abraham:murphey:cortes:karydis:2026} for surveys.
Further approaches include robust observer synthesis using linear Koopman approximations with frequency-domain error characterization~\citep{dahdah:forbes:2024}, and the construction of LPV Koopman models from noisy output data using deep state-space encoders~\citep{iacob:beintema:schoukens:toth:2021,iacob:szecsi:mate:beintema:schoukens:toth:2025}.
However, these approaches either rely on the restrictive assumption of an exactly invariant Koopman dictionary or do not provide guarantees for the designed observer.
More broadly, output-feedback control using Koopman embeddings remains largely restricted to settings that assume exact Koopman invariance~\citep{kieboom:bartzioka:jafarian:2023}, linear Koopman approximations~\citep{deutscher:2024,lopez:heinrich:muller:2026}, or LPV surrogate models with frequency-domain error characterization~\citep{eyuboglu:strasser:allgower:karimi:2026}, and a rigorous data-driven output-feedback design for nonlinear systems with provable closed-loop guarantees is missing from the literature.

To solve this gap, we propose a data-driven output-feedback controller design method for nonlinear systems.
In particular, we combine Koopman operator theory with an extended state representation of nonlinear systems based on input-output data.
This allows us to build on robust state-feedback design schemes from the literature to exponentially stabilize the nonlinear extended-state system and, thereby, establish closed-loop guarantees for the underlying nonlinear system from input-output data.
By exploiting the observability of the system, we show exponential convergence to the origin of the corresponding system's state.
The proposed approach is the first Koopman-based controller design method that relies solely on input-output data and provides closed-loop guarantees for the underlying nonlinear system, paving the way forward to rigorous data-driven output-feedback control with closed-loop guarantees.
Finally, we validate our theoretical findings in numerical simulations.

The paper is structured as follows.
In Section~\ref{sec:background}, we introduce the problem setting and necessary background for the results developed in this paper.
Section~\ref{sec:IO-representations-extended-state} is devoted to an input-output representation of nonlinear systems.
Section~\ref{sec:DD-output-feedback-controller-design} contains our main contributions, namely a data-driven output-feedback controller design for nonlinear systems on the basis of a bilinear Koopman surrogate model.
Finally, the theoretical results are illustrated in Section~\ref{sec:numerical-examples} using numerical simulations, before concluding the paper in Section~\ref{sec:conclusion}.

\section{Problem setting and background}\label{sec:background}
First, we introduce the problem setting in Section~\ref{sec:problem-setting}.
Then, we provide the necessary background on Koopman operator theory and its usage for controlled nonlinear systems in Section~\ref{sec:Koopman-background}.

\subsection{Problem setting}\label{sec:problem-setting}
We consider a nonlinear system
\begin{subequations}\label{eq:dynamics-nonlinear}
    \begin{align}
        x_{k+1} &= f(x_k,u_k), \qquad x_0=\bar{x}\in\bbR^n,\\
        y_k &= h(x_k,u_k)
    \end{align}%
\end{subequations}%
with $x\in\bbX\subseteq\bbR^n$, $u\in\bbU\subseteq\bbR^m$, and $y\in\bbY\subseteq\bbR^p$.
The state transition map $f:\bbR^n\times\bbR^m \to \bbR^n$ and the output map $h:\bbR^n\times\bbR^m \to \bbR^p$ are unknown, but assumed to be sufficiently smooth.
\begin{assumption}[Smoothness]\label{ass:smoothness}
    The maps $f,h\in C^1$, i.e., the maps are continuously differentiable, and the sets $\bbX$, $\bbU$, $\bbY$ are compact, convex, and contain the origin in its interior.
    Further, $f(0,0)=0$ and $h(0,0)=0$.
\end{assumption}
For a given initial state $\bar{x}\in\bbX$ and an input sequence $\useq_{0:L-1}\in\bbU^L$, we define the corresponding state trajectory recursively via
\begin{align}
    x_0 = \bar{x},
    \qquad 
    x_{j+1} = f(x_j,u_j), \quad j = 0,...,L-1
\end{align}%
such that $x_j$ is the state at time $j$ starting from the initial condition $\bar{x}$ at time zero. 
The corresponding output is $y_j = h(x_j,u_{j})$.
Further, we write 
\begin{equation}\label{eq:dynamics-nonlinear-j-th-step}
    f^{j}(\bar{x},\useq_{0:j-1})
    = f(x_{j-1},u_{j-1}) = x_j,
\end{equation}%
which corresponds to applying the map $f$ $j$ times.

To find a suitable system representation, the unknown system dynamics are estimated from data, where, however, the state $x$ is inaccessible and only input-output measurements are available.
In particular, we collect $d\in \bbN$ input-output trajectories of length $L+1$, i.e.,
\begin{equation}\label{eq:data-collection-IO}
    \cD = \left\{\{u_t^{(k)},y_t^{(k)}\}_{t=-L}^{0}\right\}_{k=1}^{d}
\end{equation}%
Note that we use only input-\emph{output} data instead of input-\emph{state} data.
For the Koopman surrogate established later, each trajectory of length $L+1$ will be associated with a data triplet of a delay embedding, its successor, and its input.
This yields a total of $d$ data triplets used for the data-driven surrogate characterization.
\begin{remark}
    As usual in data-driven control, we may also collect one single trajectory instead of multiple shorter input-output trajectories.
    More precisely, instead of the data $\cD$ we could also collect a trajectory of length $L+d+1$, i.e., $\{u_t,y_t\}_{t=-L}^{d}$.
    Based on this trajectory, we could again define $d$ extended state triplets, leading to the proposed Koopman-based surrogate.
    Then, all results established in this paper remain valid.
\end{remark}
To characterize when the full system behavior can be inferred from input-output data alone, we define the following.
\begin{definition}[Observability map of order $L$]\label{def:observability-map}
    For any initial condition $\bar{x}\in\bbX$ and input sequence $\useq_{0:L-1}\in\bbU^L$ for integer $L\geq 1$, we define the observability map of order $L$ as the function that stacks $L$ consecutive output values, i.e., 
    \begin{equation}
        \cO_L(\bar{x},\useq_{0:L-1}) 
        = \begin{bmatrix}
            y_{0}\\
            y_{1}\\
            y_{2}\\
            \vdots \\
            y_{L-1}\\
        \end{bmatrix}
        = \begin{bmatrix}
            h(x_{0},u_{0})\\
            h(x_{1},u_{1}) \\
            h(x_{2},u_{2}) \\
            \vdots \\
            h(x_{L-1},u_{L-1})
        \end{bmatrix}
        = \begin{bmatrix}
            h(\bar{x},u_{0}) \\
            h(f(\bar{x},u_{0}),u_{1}) \\
            h(f^{2}(\bar{x},\useq_{0:1}),u_{2}) \\
            \vdots \\
            h(f^{L-1}(\bar{x},\useq_{0:L-2}),u_{L-1})
        \end{bmatrix}.
    \end{equation}%
\end{definition}
The observability map used here coincides with the one given for uncontrolled systems in~\citet{iacob:szecsi:mate:beintema:schoukens:toth:2025}, where a Koopman model is identified from noisy input-output data and the estimation error is shown to vanish asymptotically.
However, explicit error bounds are not provided, and the convergence result assumes exponential stability of the Koopman model, which is generally difficult to satisfy even when the original system is exponentially stable~\citep[see][Appendix C]{philipp:schaller:worthmann:peitz:nuske:2025}.

Based on $\cO_L$, we can now introduce the notion of uniform observability~\citep[cf.][]{gauthier:kupka:1994,gauthier:kupka:2001,gauthier:hammouri:othman:2002,hanba:2009} relevant for the later established results.
\begin{definition}[Uniform observability]\label{def:uniform-observability}
    The system~\eqref{eq:dynamics-nonlinear} is \emph{uniformly observable} on $\bbX$ if there exists an integer $L\geq 1$ such that, for every admissible input sequence $\useq_{0:L-1}\in\bbU^L$, the map $\cO_L(\cdot,\useq_{0:L-1}): \bbX \to \bbR^{Lp}$ is injective. 
\end{definition}
The injectivity corresponding to uniform observability ensures that the initial state $\bar{x}$ can be uniquely recovered based on measured input-output sequences of length $L$, i.e., $\cO_L(x_1,\useq)=\cO_L(x_2,\useq)$ implies $x_1=x_2$.
This is a common assumption in data-driven control via input-output data~\citep[cf.][]{dai:depersis:monshizadeh:tesi:2023,dai:depersis:monshizadeh:tesi:2025}.
In the remainder of the paper, we enforce injectivity and, thus, uniform observability by a slightly stronger condition that ensures well-posedness of the inverse map $\cO_L^{-1}$.
This allows us to derive a suitable input-output characterization of the nonlinear system~\eqref{eq:dynamics-nonlinear}.
To this end, we define an extended state based on past input-output data and employ a Koopman lifting, yielding a data-driven bilinear surrogate of the underlying nonlinear system.
This representation is then used for an output-feedback controller design with closed-loop guarantees for the unknown nonlinear system.

\subsection{Koopman operator background}\label{sec:Koopman-background}
Consider a discrete-time autonomous dynamical system $
    x_{k+1} = f(x_k)
$ with $x_k \in \bbX \subseteq \bbR^n$ and state-transition map $f : \bbX \to \bbX$.
Instead of evolving the state directly, the \emph{Koopman operator} $\cK$~\citep{koopman:1931} acts on scalar-valued observable functions $\varphi : \bbX \to \bbR$ by composing them with $f$, i.e.,
\begin{equation}
    (\cK \varphi)(x_k) = \varphi\big(f(x_k)\big) 
    = \varphi(x_{k+1}).
\end{equation}%
Although the underlying dynamics $f$ may be nonlinear, $\mathcal{K}$ is a \emph{linear} operator acting on an, in general, infinite-dimensional function space.
In practice, one works with a finite-dimensional approximation by selecting a dictionary of $N$ observables, i.e.,
\begin{equation}
    \Psi(x) = \begin{bmatrix} \psi_1(x) & \cdots & \psi_N(x) \end{bmatrix}^\top \in \mathbb{R}^N,
\end{equation}%
and seeking a matrix $K \in \bbR^{N \times N}$ such that
\begin{equation}
    \Psi(x_{k+1}) \approx K \Psi(x_k).
\end{equation}
A standard data-driven method for identifying $K$ is \emph{extended dynamic mode decomposition}~\citep[EDMD][]{williams:kevrekidis:rowley:2015}.    
Given a dataset of $d+1$ consecutive state pairs $\{(x_j,\, x_{j+1})\}_{j=0}^{d}$ collected along one or more system trajectories, EDMD computes the least-squares solution
\begin{equation}
    K 
    = \argmin_{K\in \bbR^{N\times N}} \sum_{j=0}^{d-1} \|\Psi(x_{j+1}) - K \Psi(x_j)\|_2^2
    = \argmin_{K\in \bbR^{N\times N}} \|\Psi(X^+) - K \Psi(X)\|_\mathrm{F}
    = \Psi(X^+) (\Psi(X))^\dagger,
\end{equation}%
where the snapshot matrices
\begin{equation}
    \Psi(X) = \begin{bmatrix}
        \Psi(x_0) & \cdots & \Psi(x_{d-1})
    \end{bmatrix}\in\bbR^{N\times d},
    \quad
    \Psi(X^+) = \begin{bmatrix}
        \Psi(x_1) & \cdots & \Psi(x_d)
    \end{bmatrix}\in\bbR^{N\times d},
\end{equation}%
collect the lifted current and successor states, respectively, and $(\cdot)^\dagger$ denotes the Moore--Penrose pseudoinverse.
    
For a controlled discrete-time system
\begin{equation}
    x_{k+1} = f(x_k, u_k), \quad x_k \in \bbX,\; u_k \in \bbU,
\end{equation}%
the Koopman framework is extended by treating the input as a parameter of the
operator. 
Assuming $u_k$ is held constant over each sampling interval, consistent with a sample-and-hold implementation, a family of input-dependent Koopman operators $\{\cK^u\}_{u \in \bbU}$ is defined by
\begin{equation}
    (\cK^u \varphi)(x) = \varphi\big(f(x, u)\big);
\end{equation}%
compare~\citet{haseli:cortes:2026}.
Each operator $\cK^u$ remains linear in the observable space for fixed~$u$.
Introducing the same dictionary of observable functions $\Psi$ as before, one seeks a finite-dimensional approximation from data.

In linear EDMD with control~\citep{proctor:brunton:kutz:2016}, one builds a lifted \emph{linear} model
\begin{equation}\label{eq:EDMDc-linear-dynamics}
    \Psi(x_{k+1}) \approx A \Psi(x_k) + B u_k,
\end{equation}%
where the matrices $A \in \bbR^{N \times N}$ and $B \in \bbR^{N \times m}$ can be identified simultaneously via a least-squares regression over snapshot data.
Although this linear evolution of the lifted state would be desirable, as it enables the direct application of linear systems theory and control design to the originally nonlinear controlled system, it introduces fundamental approximation errors. 
In particular, representations of the above form impose severe limitations on the original system~\citep{shang:haseli:cortes:zheng:2026,heeg:worthmann:2026}.
This fact prevents closed-loop guarantees and thereby safe control of the underlying nonlinear system~\citep{iacob:toth:schoukens:2024,strasser:worthmann:mezic:berberich:schaller:allgower:2026}.
Instead, at least \emph{bilinear} Koopman surrogates are required to approximate the infinite-dimensional Koopman action.
In particular, the Koopman-based control framework proposed in~\citet{strasser:2026} leverages the bilinear Koopman surrogate modeling approaches SafEDMD~\citep{strasser:schaller:worthmann:berberich:allgower:2025,strasser:schaller:worthmann:berberich:allgower:2026} and kEDMD~\citep{strasser:schaller:berberich:worthmann:allgower:2025} to ensure closed-loop properties of the underlying nonlinear system~\citep{strasser:berberich:schaller:worthmann:allgower:2025}.
More precisely, the resulting bilinear Koopman surrogate reads
\begin{equation}\label{eq:bilinear-surrogate-literature-state}
    \Psi(x_{k+1}) = A \Psi(x_k) + B_0 u_k + \tilde{B} (u_k\otimes \Psi(x_k)) + r(x_k,u_k), 
\end{equation}%
where the residual can be \emph{proportionally} bounded by
\begin{equation}
    \|r(x_k,u_k)\| \leq c_x \|\Psi(x_k)\| + c_u \|u_k\|, \qquad c_x,c_u\geq 0,
\end{equation}%
and the matrices $A\in\bbR^{N\times N}$, $B_0\in\bbR^{N\times m}$, and $\tB\in\bbR^{N\times Nm}$ are either computed via kernel methods or least-squares regression over snapshot data, i.e.,
\begin{equation}
    \begin{bmatrix}
        A & B_0 & \tB
    \end{bmatrix} 
    = \Psi(X^+) \begin{bmatrix} \Psi(X) \\ U \\ U_X \end{bmatrix}^\dagger,
\end{equation}%
where $U = \begin{bmatrix} u_0 & \cdots & u_{d-1} \end{bmatrix}$ and $U_X = \begin{bmatrix} u_0\otimes \Psi(x_0) & \cdots & u_{d-1}\otimes \Psi(x_{d-1}) \end{bmatrix}$.

The main limitation of the bilinear surrogate model~\eqref{eq:bilinear-surrogate-literature-state} is its dependence on the state $x$, which is typically not accessible in practice and therefore restrictive. 
This manifests in two ways: first, learning the surrogate model requires state measurements; second, the resulting state-space representation relies on the state for both prediction and controller design.
In the remainder of this paper, we circumvent these issues by generalizing the state-dependent Koopman surrogate models introduced above, deriving an input-output characterization that is independent of the unknown state $x$.

\section{Input-output representation of nonlinear systems}\label{sec:IO-representations-extended-state}
In this section, we discuss input-output representations of nonlinear systems, extending the ideas of linear systems, which we recall in Appendix~\ref{app:IO-representation-linear} for completeness.
We note that the results developed in this section may be of independent interest beyond Koopman-based control; see Remark~\ref{rem:IO-interest-beyond-Koopman}.

Consider the nonlinear system~\eqref{eq:dynamics-nonlinear} with $x\in\bbX\subseteq\bbR^n$, $u\in\bbU\subseteq\bbR^m$, and $y\in\bbY\subseteq\bbR^p$, where the state $x$ is inaccessible but only input-output measurements are available.
Since the original internal state $x$ is unknown, the core idea is to use past input-output measurements to infer knowledge about the initial condition.
In particular, we seek an equivalent input-output representation of the nonlinear system~\eqref{eq:dynamics-nonlinear} based on an extended state of delayed input-output measurements.
To this end, we construct a delay embedding $\xi$ of past input-output measurements over a finite horizon $L$, which we call the extended-state vector, as common in the data-driven control literature.
More precisely, we define
\begin{equation}\label{eq:extended-state}
    \xi_k = \begin{bmatrix}
        u_{k-L}^\top & u_{k-L+1}^\top & \cdots & u_{k-1}^\top &
        y_{k-L}^\top & y_{k-L+1}^\top & \cdots & y_{k-1}^\top
    \end{bmatrix}^\top.
\end{equation}%

In the following, we investigate under which conditions there exists a smooth map $\cR_L$ such that the initial state $\bar{x}$ is uniquely reconstructable from a finite input-output sequence. 
In particular, we aim at the representation~\citep{isidori:1985}
\begin{equation}\label{eq:reconstruction-x-from-u-y}
    x_{k-L} = \cR_L(\useq_{k-L:k-1},\yseq_{k-L:k-1});
\end{equation}%
compare the reconstructability map in~\citet{iacob:szecsi:mate:beintema:schoukens:toth:2025}.
If such $\cR_L$ exists, we can substitute it into the combined dynamics and output equation to obtain an equivalent closed-form input-output relation that does not rely on the original state $x_{k-L}$.
This is exactly characterized by uniform observability of the underlying nonlinear system~\eqref{eq:dynamics-nonlinear}, compare Definition~\ref{def:uniform-observability}.
In the following, we characterize under which conditions this observability property holds and, thereby, an equivalent input-output representation exists.

Based on Assumption~\ref{ass:smoothness}, we directly deduce $\cO_L\in C^1$, i.e., we can compute the Jacobian of $\cO_L$ w.r.t. $\bar{x}$.
In particular,
\begin{equation}
    \frac{\partial \cO_L}{\partial \bar{x}}
    = \begin{bmatrix}
        H_0
        \\
        H_1
        F_0
        \\
        H_2
        F_1
        F_0
        \\
        \vdots 
        \\
        H_{L-1}
        F_{L-2}
        \cdots  
        F_1 
        F_0
    \end{bmatrix},
    \qquad
    \begin{aligned}
        H_j 
        &= \frac{\partial h}{\partial x}(x_j,u_{j}) 
        = \frac{\partial h}{\partial x}(f^j(\bar{x},\useq_{0:j-1}),u_{j}),
        \\
        F_j 
        &= \frac{\partial f}{\partial x}(x_j,u_{j}) 
        = \frac{\partial f}{\partial x}(f^j(\bar{x},\useq_{0:j-1}),u_{j}),
    \end{aligned}%
\end{equation}%
where $H_j$ and $F_j$ depend on the trajectory $x_j$, i.e., on both the initial condition $\bar{x}$ and the input sequence $\useq_{0:j}$.
\begin{assumption}[Uniform injectivity of the observability map]\label{ass:lower-bound-Jacobian-observability-map}
    There exists an integer $L\geq 1$ and a threshold $\alpha>0$ such that
    \begin{equation}\label{eq:lower-bound-Jacobian-observability-map}
        \|\cO_L(x_1,\useq_{0:L-1}) - \cO_L(x_2,\useq_{0:L-1})\|
        \geq 
        \alpha \|x_1 - x_2\|
    \end{equation}%
    holds for all $x_1,x_2\in\bbX$ and all $\useq_{0:L-1}\in\bbU^L$.
\end{assumption}
Assumption~\ref{ass:lower-bound-Jacobian-observability-map} can be understood as a quantitative injectivity ensuring observability of the underlying system.
A \emph{pointwise} condition that is typically easier to verify in practice compared to Assumption~\ref{ass:lower-bound-Jacobian-observability-map} is
\begin{equation}
    \sigma_{\min}\!\left(\frac{\partial \cO_L}{\partial \bar{x}}(\bar{x},\useq_{0:L-1})\right) \geq \alpha
\end{equation}%
for all $\bar{x}\in\bbX$ and all $\useq_{0:L-1}\in\bbU^L$, where $\sigma_{\min}$ denotes the smallest singular value.
However, this pointwise condition does not imply~\eqref{eq:lower-bound-Jacobian-observability-map} in general, and the stronger Assumption~\ref{ass:lower-bound-Jacobian-observability-map} is required for the following result.
There, we explain a given input-output trajectory of~\eqref{eq:dynamics-nonlinear} by an extended-state system if there is a bijection between the given trajectory and the extended-state trajectory.
\begin{proposition}[Input-output representation of nonlinear systems]\label{prop:IO-representation-nonlinear}
    Suppose Assumption~\ref{ass:smoothness} and Assumption~\ref{ass:lower-bound-Jacobian-observability-map} hold.
    Then, any trajectory $\{u_k,y_k\}_{k=-L}^{d}$ of~\eqref{eq:dynamics-nonlinear} can be explained by the extended-state system 
    \begin{subequations}\label{eq:IO-representation-nonlinear}
        \begin{align}
            \label{eq:IO-representation-nonlinear-state}
            \xi_{k+1} &= f_\xi(\xi_k,u_k), \\
            y_k &= h_\xi(\xi_k,u_k)
            \label{eq:IO-representation-nonlinear-output}
        \end{align}%
    \end{subequations}%
    with an extended state~\eqref{eq:extended-state} initialized as
    \begin{equation}\label{eq:extended-state-initialization}
        \xi_0 = \begin{bmatrix}
            u_{-L}^\top & u_{-L+1}^\top & \cdots & u_{-1}^\top & 
            y_{-L}^\top & y_{-L+1}^\top & \cdots & y_{-1}^\top
        \end{bmatrix}^\top\in\bbR^{L(p+m)},
    \end{equation}%
    where
    \begin{subequations}
        \begin{align}
            f_\xi(\xi_k,u_k) 
            &= \left[\begin{array}{cccc|cccc}
                0 & I & \cdots & 0 & 0 & \cdots & 0 & 0 \\
                \vdots & \vdots & \ddots & \vdots & \vdots & \ddots & \vdots & \vdots \\
                0 & 0 & \cdots & I & 0 & \cdots & 0 & 0 \\
                0 & 0 & \cdots & 0 & 0 & \cdots & 0 & 0 \\
                \hline
                0 & 0 & \cdots & 0 & 0 & I & \cdots & 0 \\
                \vdots & \vdots & \ddots & \vdots & \vdots & \ddots & \vdots & \vdots \\
                0 & 0 & \cdots & 0 & 0 & 0 & \cdots & I \\
                0 & 0 & \cdots & 0 & 0 & 0 & \cdots & 0
            \end{array}\right]
            \xi_k
            + \left[\begin{array}{c}
                0 \\
                \vdots \\
                0 \\
                I_m \\\hline
                0 \\
                \vdots \\
                0 \\
                0
            \end{array}\right]
            u_k
            + \left[\begin{array}{c}
                0 \\
                \vdots \\
                0 \\
                0 \\\hline
                0 \\
                \vdots \\
                0 \\
                h(f^L(\cO_L^{-1}(\xi_{k}),\xi_k),u_{k})
            \end{array}\right],
            \\
            h_\xi(\xi_k,u_k)
            &= \left[\begin{array}{cccc|cccc}
                0 & 0 & \cdots & 0 & 0 & \cdots & 0 & I_p 
            \end{array}\right]
            f_\xi(\xi_k,u_k)
        \end{align}%
    \end{subequations}%
    using 
    \begin{subequations}
        \begin{align}
            f^L(\cO_L^{-1}(\xi_{k}),\xi_k)
            &= f^L(\cO_L^{-1}(\xi_{k}),\begin{bmatrix}I_{mL} & 0_{mL\times pL}\end{bmatrix}\xi_k)
            = f^L(\cO_L^{-1}(\xi_{k}),\useq_{k-L:k-1}),
            \\
            \cO_L^{-1}(\xi_{k})
            &= \cO_L^{-1}(\begin{bmatrix}I_{mL} & 0_{mL\times pL}\end{bmatrix}\xi_k,\begin{bmatrix}0_{pL\times mL} & I_{pL}\end{bmatrix}\xi_k)
            = \cO_L^{-1}(\useq_{k-L:k-1},\yseq_{k-L:k-1}).
        \end{align}%
    \end{subequations}%
\end{proposition}
\begin{proof}
    The structure in $f_\xi$ due to the shifts of inputs and outputs results using similar arguments as in the linear case; see Appendix~\ref{app:IO-representation-linear}.
    It remains to show the last row of $f_\xi$, i.e., the dynamics for $y_k$.

    First, we establish the reconstruction of the initial condition $x_{k-L}$.
    Assumption~\ref{ass:lower-bound-Jacobian-observability-map} directly implies that $\cO_L(\cdot,\useq_{0:L-1})$ is injective on $\bbX$ for every $\useq_{0:L-1}\in\bbU^L$.
    In particular, if $\cO_L(x_1,\useq_{0:L-1})=\cO_L(x_2,\useq_{0:L-1})$, 
    then~\eqref{eq:lower-bound-Jacobian-observability-map} gives $\alpha\|x_1-x_2\|\leq 0$
    and, since $\alpha>0$, we conclude $x_1=x_2$.    
    Hence, $\cO_L(\cdot,\useq_{0:L-1})$ is uniformly injective on $\bbX$, which corresponds to uniform observability of the underlying system according to Definition~\ref{def:uniform-observability}.
    Thus, the initial state can be uniquely reconstructed.
    More precisely, for each $\useq_{0:L-1}\in\bbU^L$ there exists a map $\cR_L(\useq_{0:L-1},\cdot): \bbR^{Lp}\to\bbX$ with $\cR_L = \cO_L^{-1}\in C^1$ such that
    \begin{equation}
        \cO_L(\bar{x},\useq_{0:L-1}) 
        = \zeta
        \qquad\iff\qquad
        \bar{x} = \cR_L(\useq_{0:L-1},\zeta).
    \end{equation}%
    Applied along a trajectory at time $k$ with $\bar{x}=x_{k-L}$, $\zeta = \yseq_{k-L:k-1}$ and $\useq_{0:L-1}=\useq_{k-L:k-1}$, we obtain the state reconstruction formula~\eqref{eq:reconstruction-x-from-u-y}.
    Substituting this into the dynamics~\eqref{eq:dynamics-nonlinear} with~\eqref{eq:dynamics-nonlinear-j-th-step} yields
    \begin{equation}
        x_{k} 
        = f^L(x_{k-L},\useq_{k-L:k-1}) 
        = f^L(\cR_L(\useq_{k-L:k-1},\yseq_{k-L:k-1}),\useq_{k-L:k-1})
    \end{equation}%
    and, thus, 
    \begin{equation}
        y_{k} = h(f^L(\cR_L(\useq_{k-L:k-1},\yseq_{k-L:k-1}),\useq_{k-L:k-1}),u_{k}).
    \end{equation}%
    Finally, we replace $\useq_{k-L:k-1}$ and $\yseq_{k-L:k-1}$ by the respective entries in $\xi_k$ and extract the last row for the definition of $h_\xi$.
    This concludes the proof.
\end{proof}
Proposition~\ref{prop:IO-representation-nonlinear} establishes a (non-minimal) input-output representation based on the extended state~\eqref{eq:extended-state}, that has the same input-output behavior as the unknown nonlinear system~\eqref{eq:dynamics-nonlinear}.
Although this representation is a natural nonlinear extension of linear subspace identification, the explicit structural input-output representation in~\eqref{eq:IO-representation-nonlinear} is hardly exploited in the literature.
More precisely, most of the available results for nonlinear subspace identification rely on linear-like approximations of the underlying nonlinear system~\citep[compare, e.g.,][]{bamieh:giarre:2002,verdult:2002,williams:kevrekidis:rowley:2015}.
Instead, we follow a different approach and use the inverse map $\cO_L^{-1}$ to characterize the initial condition, exploiting uniform observability of the underlying system.

More precisely, Assumption~\ref{ass:lower-bound-Jacobian-observability-map} enforces that the inverse map $\cO_L^{-1}(\useq_{0:L-1},\cdot)$ is well-defined on $\cO_L(\bbX,\useq_{0:L-1})$ for every $\useq_{0:L-1}\in\bbU^L$, since the uniform expansion condition~\eqref{eq:lower-bound-Jacobian-observability-map} directly implies injectivity of $\cO_L(\cdot,\useq_{0:L-1})$ on $\bbX$, uniformly in $\useq_{0:L-1}$.
In addition, Assumption~\ref{ass:smoothness} ensures boundedness of $\frac{\partial \cO_L}{\partial \bar{x}}$ on $\bbX\times\bbU^L$, i.e., $\|\frac{\partial \cO_L}{\partial \bar{x}}(\bar{x},\useq_{0:L-1})\| \leq L_{\cO} < \infty$ for all $\bar{x}\in\bbX$ and $\useq_{0:L-1}\in\bbU^L$, which follows from continuity of $\frac{\partial\cO_L}{\partial\bar{x}}$ and compactness of $\bbX\times\bbU^L$.
Combining this upper bound with the lower bound in Assumption~\ref{ass:lower-bound-Jacobian-observability-map} yields the bi-Lipschitz condition
\begin{equation}\label{eq:bi-Lipschitz-cO}
    \alpha \| x_1 - x_2 \|
    \leq
    \| \cO_L(x_1,\useq_{0:L-1}) - \cO_L(x_2,\useq_{0:L-1}) \|
    \leq
    L_{\cO} \| x_1 - x_2\|
\end{equation}%
uniformly in $\useq_{0:L-1}\in\bbU^L$, where the upper inequality follows from the mean value inequality for vector-valued $C^1$-maps on the convex set $\bbX$~\citep{rudin:1976}.
Inverting the lower bound in~\eqref{eq:bi-Lipschitz-cO} directly yields
\begin{equation}\label{eq:Lipschitz-inverse}
    \| \cO_L^{-1}(\useq_{0:L-1},y_1) - \cO_L^{-1}(\useq_{0:L-1},y_2) \|
    \leq 
    \frac{1}{\alpha} \| y_1 - y_2 \|,
\end{equation}%
i.e., the inverse map is Lipschitz continuous uniformly in $\useq_{0:L-1}\in\bbU^L$, which prevents noise amplification in the presence of small measurement errors.
Further, since $\sigma_{\min}(\frac{\partial\cO_L}{\partial\bar{x}})\geq\alpha>0$ uniformly on $\bbX\times\bbU^L$, which is implied by~\eqref{eq:lower-bound-Jacobian-observability-map} via differentiation, the inverse function theorem~\citep{rudin:1976} guarantees that $\cO_L^{-1}\in C^1$, i.e., the inverse map is continuously differentiable.
Hence, the initial state~$\bar{x}$ can be uniquely and smoothly reconstructed from input-output sequences $\useq_{0:L-1}$, $\yseq_{0:L-1}$, and the extended system representation in~\eqref{eq:IO-representation-nonlinear} is well-defined.
\begin{remark}
    Takens' theorem~\citep{takens:1981,sauer:yorke:casdagli:1991} analyzes the delay-coordinate map~$\cO_L$ defined in Definition~\ref{def:observability-map} in the autonomous case (i.e., without control input $u$). 
    It states that, for generic systems and output maps, the map $\cO_L$ generically becomes an embedding for $L\geq 2n$.
    In contrast, uniform observability according to Definition~\ref{def:uniform-observability} requires injectivity for a fixed system and uniformly over all input sequences $\useq_{0:L-1}\in\bbU^L$, while the lower bound in Assumption~\ref{ass:lower-bound-Jacobian-observability-map} further strengthens this to a quantitative (well-conditioned) embedding, which is not guaranteed by Takens.
\end{remark}
\begin{remark}\label{rem:IO-interest-beyond-Koopman}
    In the remainder of the paper, we use the input-output representation established by Proposition~\ref{prop:IO-representation-nonlinear} to derive a Koopman-based output-feedback controller design method. 
    We note, however, that Proposition~\ref{prop:IO-representation-nonlinear} is not specific to Koopman-based methods and could serve as a foundation for generalizing other nonlinear data-driven state-feedback designs to the output-feedback setting; see~\citet{martin:schon:allgower:2023b} for an overview of such designs with closed-loop guarantees. 
    For example, one could define a set of basis functions and apply a nonlinear data-driven control approach as in~\citet{lazar:2024}, or construct a polynomial approximation of the nonlinear input-output dynamics in the extended state $\xi$ and combine it with robust control and sum-of-squares (SOS) optimization as in~\citet{martin:allgower:2024,martin:2024}. 
    In the present work, Koopman is a particularly natural choice given that the nonlinear dynamics~\eqref{eq:dynamics-nonlinear} and the associated maps $f_\xi$, $h_\xi$ are unknown.
\end{remark}

\section{Data-driven output-feedback controller design for nonlinear systems}\label{sec:DD-output-feedback-controller-design}
After establishing an input-output representation of unknown nonlinear systems, we leverage Koopman operator theory to design an output-feedback controller for nonlinear systems with rigorous closed-loop guarantees.
To this end, Section~\ref{sec:IO-Koopman-surrogate} is devoted to deriving a bilinear surrogate model for the nonlinear input-output representation of the underlying system.
Then, we use this surrogate to design an output-feedback controller in Section~\ref{sec:IO-Koopman-controller-design}.

\subsection{Koopman-based bilinear surrogate of nonlinear input-output behavior}\label{sec:IO-Koopman-surrogate}
Based on the nonlinear input-output representation~\eqref{eq:IO-representation-nonlinear} of the underlying nonlinear system~\eqref{eq:dynamics-nonlinear} established by Proposition~\ref{prop:IO-representation-nonlinear}, we employ a Koopman lifting to bilinearize the representation, which is subsequently learned via data.
To this end, we follow the line of thoughts presented in~\citet{strasser:schaller:worthmann:berberich:allgower:2026,strasser:worthmann:mezic:berberich:schaller:allgower:2026} and introduce the vector-valued observable function $\Psi: \bbR^{L(m+p)}\to\bbR^{N}$ with
\begin{equation}
    \Psi(\xi) = \begin{bmatrix}
        \xi^\top & \psi_{L(m+p)+1} & \cdots & \psi_N
    \end{bmatrix}.
\end{equation}%
The observables $\psi_k$, $k=L(m+p)+1,...,N$, satisfy $\psi_k\in C^1(\bbR^{L(m+p)},\bbR)$ with $\psi_k(0)=0$.
Then, $\Psi$ is pointwise bounded by
\begin{equation}
    \|\xi\|_2 
    \leq \| \Psi(\xi) \|_2 
    \leq L_\Psi \|\xi\|_2
\end{equation}%
for all $\xi\in\bbE\coloneqq \bbU^L \times \bbY^L\subseteq\bbR^{L(m+p)}$ with some $L_\Psi\in\bbR$~\citep[Section 2.1]{strasser:schaller:worthmann:berberich:allgower:2026}.

Since the nonlinear system~\eqref{eq:dynamics-nonlinear} and its input-output representation~\eqref{eq:IO-representation-nonlinear} are unknown, we characterize the dynamics via data.
In particular, we collect input-output data $\cD$ consisting of $d$ trajectories of length $L+1$ as defined in~\eqref{eq:data-collection-IO}.
This trajectory allows us to arrange the data according to the defined extended state~\eqref{eq:extended-state} for delay depth $L$, i.e., we obtain the extended-state data 
\begin{equation}\label{eq:data-extended-state}
    \cD = \{(\xi_k,\xi_k^+,u_k)\}_{k=1}^{d}
\end{equation}%
consisting of $d$ triplets of extended state, its successor, and its input, where $u_k=u_0^{(k)}$.

In the following, we present two different data-driven surrogate representations derived using the data $\cD$. 
In Section~\ref{sec:IO-Koopman-surrogate-nominal}, we first assume the existence of an \emph{exact} Koopman bilinearization.
Since this is typically not the case for general nonlinear systems, we allow for perturbations in the Koopman-based bilinear surrogate in Section~\ref{sec:IO-Koopman-surrogate-perturbed}.

\subsubsection{Exact Koopman bilinearization}\label{sec:IO-Koopman-surrogate-nominal}
The first presented data-driven surrogate model relies on the following (possibly restrictive) assumption on the employed Koopman bilinearization, which allows for the use of straightforward arguments.
We stress, however, that the main results developed in this paper do not rely on this assumption, and the subsequent section considers a realistic and more general setting.
\begin{assumption}\label{ass:Koopman-bilinearization-exact}
    The Koopman operator action corresponding to the nonlinear input-output representation~\eqref{eq:IO-representation-nonlinear} admits an \emph{exact} finite-dimensional bilinear representation of the form
    \begin{equation}\label{eq:Koopman-bilinearization-exact}
        \Psi(\xi_{k+1}) = A_\mathrm{tr} \Psi(\xi_k) + B_{0,\mathrm{tr}} u_k + \tB_\mathrm{tr} (u_k\otimes \Psi(\xi_k)).
    \end{equation}%
    for all $\xi\in\bbE$ and $u\in\bbU$.
\end{assumption}
The unknown matrices $A_\mathrm{tr},B_{0,\mathrm{tr}},\tB_\mathrm{tr}$ of the bilinear Koopman representation~\eqref{eq:Koopman-bilinearization-exact} are estimated from measured data.
To this end, we solve the linear regression problem
\begin{equation}\label{eq:regression-bilinear-EDMD-nominal}
    \min_{A,B_0,\tB} 
    \left\|
        \Psi(\Xi^+) - A \Psi(\Xi) - B_0 U - \tB U_\Xi 
    \right\|_\mathrm{F}
\end{equation}%
with the data matrices
\begin{subequations}\label{eq:data-matrices-IO-regression}
    \begin{alignat}{2}
        \Psi(\Xi) &= \begin{bmatrix}
            \Psi(\xi_1) & \cdots & \Psi(\xi_{d})
        \end{bmatrix},
        &\qquad
        U &= \begin{bmatrix}
            u_1 & \cdots & u_{d}
        \end{bmatrix},
        \\
        \Psi(\Xi^+) &= \begin{bmatrix}
            \Psi(\xi_1^+) & \cdots & \Psi(\xi_d^+)
        \end{bmatrix},
        &\qquad
        U_\Xi &= \begin{bmatrix}
            u_1\otimes \Psi(\xi_1) & \cdots & u_{d} \otimes \Psi(\xi_{d})
        \end{bmatrix}.
    \end{alignat}%
\end{subequations}
The regression problem~\eqref{eq:regression-bilinear-EDMD-nominal} has a unique solution if 
\begin{equation}\label{eq:rank-condition-linear-regression}
    \rank \begin{bmatrix}
        \Psi(\Xi) \\
        U \\
        U_\Xi
    \end{bmatrix}
    = N + m + mN,
\end{equation}%
i.e., the stacked data matrix has full column rank.
Thus, we obtain
\begin{equation}\label{eq:least-squares-solution}
    \begin{bmatrix}
        A & B_0 & \tB
    \end{bmatrix}
    = \Psi(\Xi^+) \begin{bmatrix}
        \Psi(\Xi) \\
        U \\
        U_\Xi
    \end{bmatrix}^\dagger
    = \Psi(\Xi^+) \begin{bmatrix}
        \Psi(\Xi) \\
        U \\
        U_\Xi
    \end{bmatrix}^\top \left(
        \begin{bmatrix}
            \Psi(\Xi) \\
            U \\
            U_\Xi
        \end{bmatrix}
        \begin{bmatrix}
            \Psi(\Xi) \\
            U \\
            U_\Xi
        \end{bmatrix}^\top
    \right)^{-1}.
\end{equation}%
This leads to the following intermediate result.
\begin{proposition}\label{prop:Koopman-bilinearization-exact}
    Let Assumption~\ref{ass:Koopman-bilinearization-exact} and the rank condition~\eqref{eq:rank-condition-linear-regression} hold. 
    Then, any trajectory of the nonlinear input-output representation~\eqref{eq:IO-representation-nonlinear} satisfies
    \begin{subequations}\label{eq:bilinear-surrogate-nominal}
        \begin{align}
            \Psi(\xi_{k+1}) 
            &= A \Psi(\xi_k) + B_{0} u_k + \tB (u_k\otimes \Psi(\xi_k)) 
            \label{eq:bilinear-surrogate-nominal-dynamics}
            \\
            y_k 
            &= \begin{bmatrix}
                0_{p\times L(m+p-1)} & I_{p} & 0_{p\times N-L(m+p)}
            \end{bmatrix}
            \Psi(\xi_{k+1})
            \label{eq:bilinear-surrogate-nominal-output}
        \end{align}%
    \end{subequations}%
    with $A,B_0,\tilde B$ in~\eqref{eq:least-squares-solution} for all $\xi\in\bbE$ and $u\in\bbU$.
\end{proposition}%
\begin{proof}
    If the data is chosen such that the rank condition~\eqref{eq:rank-condition-linear-regression} holds, then the learning error is zero, i.e., the least-squares solution~\eqref{eq:least-squares-solution} corresponds to the true solution, i.e., $
        \begin{bmatrix}
            A & B_0 & \tB
        \end{bmatrix}
        = \begin{bmatrix}
            A_\mathrm{tr} & B_{0,\mathrm{tr}} & \tB_\mathrm{tr}
        \end{bmatrix}
    $. 
    Then, leveraging Assumption~\ref{ass:Koopman-bilinearization-exact} establishes~\eqref{eq:bilinear-surrogate-nominal-dynamics}.
    The output equation~\eqref{eq:bilinear-surrogate-nominal-output} directly follows from the nonlinear output equation~\eqref{eq:IO-representation-nonlinear-output}, i.e., 
    \begin{align}
        y_k 
        &= \left[\begin{array}{cccc|cccc}
            0 & 0 & \cdots & 0 & 0 & 0 & \cdots & I_p 
        \end{array}\right] 
        \xi_{k+1}
        \nonumber\\
        &= \left[\begin{array}{cccc|cccc}
            0 & 0 & \cdots & 0 & 0 & 0 & \cdots & I_p 
        \end{array}\right] 
        \begin{bmatrix}
            I_{L(m+p)} & 0_{L(m+p)\times N-L(m+p)}
        \end{bmatrix}
        \Psi(\xi_{k+1}).
    \end{align}%
    This establishes~\eqref{eq:bilinear-surrogate-nominal-output} and, thus, concludes the proof.    
\end{proof}
According to Proposition~\ref{prop:Koopman-bilinearization-exact}, we can design an extended-state-feedback controller for the nonlinear input-output representation~\eqref{eq:IO-representation-nonlinear} based on the finite-dimensional error-free bilinear surrogate representation~\eqref{eq:bilinear-surrogate-nominal}.
Importantly, an extended-state-feedback controller for~\eqref{eq:IO-representation-nonlinear} corresponds to an output-feedback controller for the underlying nonlinear system~\eqref{eq:dynamics-nonlinear}.

\subsubsection{Koopman bilinearization with proportionally bounded residual}\label{sec:IO-Koopman-surrogate-perturbed}
Since Assumption~\ref{ass:Koopman-bilinearization-exact} is typically hard to satisfy for general nonlinear systems, we loosen the assumption and allow for errors in the bilinear representation satisfying a proportional error bound.
\begin{assumption}\label{ass:Koopman-bilinearization-perturbed}
    The Koopman operator action corresponding to the nonlinear input-output representation~\eqref{eq:IO-representation-nonlinear} admits a finite-dimensional perturbed bilinear representation of the form
    \begin{equation}\label{eq:Koopman-bilinearization-perturbed}
        \Psi(\xi_{k+1}) = A_\mathrm{tr} \Psi(\xi_k) + B_{0,\mathrm{tr}} u_k + \tB_\mathrm{tr} (u_k\otimes \Psi(\xi_k)) + r_\Psi(\xi_k,u_k),
    \end{equation}%
    where the residual $r_\Psi(\xi_k,u_k)$ is proportionally bounded by
    \begin{equation}\label{eq:Koopman-bilinearization-proportional-error-bound}
        \|r_\Psi(\xi,u)\| \leq c_{\Psi,\xi} \|\Psi(\xi)\| + c_{\Psi,u} \|u\|, \qquad c_{\Psi,\xi},c_{\Psi,u} \geq 0
    \end{equation}%
    for all $\xi\in\bbE$ and $u\in\bbU$.
\end{assumption}
The proportional structure of bound~\eqref{eq:Koopman-bilinearization-proportional-error-bound} on the residual error is validated in~\citet[Theorem~3.1, Corollary~3.2]{strasser:schaller:worthmann:berberich:allgower:2026} and represents a standard assumption in Koopman-based control~\citep{strasser:worthmann:mezic:berberich:schaller:allgower:2026}.
This bound captures both the bilinear approximation error of the Koopman operator and the projection error arising from the finite-dimensional lifting function $\Psi$.
For the latter, the proportional structure is widely assumed in the literature~\citep{strasser:schaller:worthmann:berberich:allgower:2026} and explicitly verified, e.g., in the kernel setting~\citep{strasser:schaller:berberich:worthmann:allgower:2025}. 
Moreover,~\citet[Lemma~A.1]{strasser:schaller:worthmann:berberich:allgower:2026} establishes a proportional bound on the Koopman-bilinearization error for control-affine nonlinear systems.

As before, we estimate the unknown system matrices $A_\mathrm{tr},B_{0,\mathrm{tr}},\tB_\mathrm{tr}$ using the collected data, i.e., by solving for~\eqref{eq:least-squares-solution} if the rank condition~\eqref{eq:rank-condition-linear-regression} is satisfied.
This leads us to the first main result of this paper.
\begin{theorem}\label{thm:Koopman-bilinearization-perturbed}
    Let Assumption~\ref{ass:Koopman-bilinearization-perturbed} and the rank condition~\eqref{eq:rank-condition-linear-regression} hold. 
    Then, any trajectory of the nonlinear input-output representation~\eqref{eq:IO-representation-nonlinear} satisfies
    \begin{subequations}\label{eq:bilinear-surrogate-perturbed}
        \begin{align}
            \Psi(\xi_{k+1}) 
            &= A \Psi(\xi_k) + B_{0} u_k + \tB (u_k\otimes \Psi(\xi_k)) 
            + r_\Psi(\xi,u) + r_\Delta(\xi,u)
            \label{eq:bilinear-surrogate-perturbed-dynamics}
            \\
            y_k 
            &= \begin{bmatrix}
                0_{p\times L(m+p-1)} & I_{p} & 0_{p\times N-L(m+p)}
            \end{bmatrix}
            \Psi(\xi_{k+1}),
            \label{eq:bilinear-surrogate-perturbed-output}
        \end{align}%
    \end{subequations}%
    where $A,B_0,\tilde B$ solve the regression problem~\eqref{eq:regression-bilinear-EDMD-nominal} and where the residuals $r_\Psi(\xi,u)$ and $r_\Delta(\xi,u)$ are bounded by~\eqref{eq:Koopman-bilinearization-proportional-error-bound} and 
    \begin{equation}\label{eq:bound-r-Delta}
        \|r_\Delta(\xi,u)\|^2 \leq c_{\Delta,\xi}^2 \|\Psi(\xi)\|^2 + c_{\Delta,u}^2 \|u\|^2
    \end{equation}%
    with 
    \begin{equation}\label{eq:bound-r-Delta-constants}
        c_{\Delta,\xi} 
        = c_{\Delta,u} \sqrt{1 + \max_{u\in\bbU}\|u\|^2}
        ,\qquad
        c_{\Delta,u} 
        = \frac{
            \sqrt{
                2c_{\Psi,\xi}^2 \|\Psi(\Xi)\|_\mathrm{F}^2 + 2c_{\Psi,u}^2 \|U\|_\mathrm{F}^2
            }
        }{
            \sigma_{\min}\left(
                \begin{bmatrix}
                    \Psi(\Xi)\\U\\U_\Xi
                \end{bmatrix}
            \right)
        },
    \end{equation}%
    respectively, for all $\xi\in\bbE$ and $u\in\bbU$.
\end{theorem}%
\begin{proof}
    Due to the perturbation $r_\Psi$ in~\eqref{eq:Koopman-bilinearization-perturbed}, the least-squares solution $(A,B_0,\tB)$ does not necessarily align with the true solution $(A_\mathrm{tr},B_{0,\mathrm{tr}},\tB_\mathrm{tr})$.
    Thus, we need to characterize the mismatch between the true values and the least-squares estimate.
    For the sake of compactness, we introduce the abbreviation 
    \begin{equation}
        R_\Psi(\Xi,U) 
        = \begin{bmatrix}
            r_\Psi(\xi_1,u_1) & \cdots & r_\Psi(\xi_{d},u_{d})
        \end{bmatrix}
    \end{equation}%
    in the following.
    
    First, we exploit results from least-squares optimization~\citep[cf.][]{ziemann:tsiamis:lee:jedra:matni:pappas:2023} to derive the upper bound
    \begin{equation}\label{eq:proof-bound-Delta-literature}
        \left\|\begin{bmatrix}
            \Delta A & \Delta B_0 & \Delta \tB
        \end{bmatrix}\right\|
        \coloneqq
        \left\|
            \begin{bmatrix} A_\mathrm{tr} & B_{0,\mathrm{tr}} & \tB_\mathrm{tr} \end{bmatrix}
            - \begin{bmatrix} A & B_0 & \tB \end{bmatrix}
        \right\|
        \leq \frac{
            \sigma_{\max}(R_\Psi(\Xi,U)
        }{
            \sigma_{\min}\left(
                \begin{bmatrix}
                    \Psi(\Xi)\\U\\U_\Xi
                \end{bmatrix}
            \right)
        }.
    \end{equation}%
    This estimate is further refined by observing
    \begin{equation}
        \sigma_{\max}(R_\Psi(\Xi,U))
        = \|R_\Psi(\Xi,U)\|_2
        \leq \|R_\Psi(\Xi,U)\|_\mathrm{F}
        = \sqrt{
            \sum_{k=1}^{d} \|r_\Psi(\xi_k,u_k)\|_2^2
        }.
    \end{equation}%
    Then, using the proportional error bound~\eqref{eq:Koopman-bilinearization-proportional-error-bound}, we obtain
    \begin{align}
        \sigma_{\max}(R_\Psi(\Xi,U))
        &\leq \sqrt{
            \sum_{k=1}^{d} \left(c_{\Psi,\xi}\|\Psi(\xi_k)\|_2 + c_{\Psi,u} \|u_k\|_2\right)^2
        }
        \leq \sqrt{
            \sum_{k=1}^{d} 2c_{\Psi,\xi}^2\|\Psi(\xi_k)\|_2^2 + 2c_{\Psi,u}^2 \|u_k\|_2^2
        }
        \nonumber\\
        &= \sqrt{
            2c_{\Psi,\xi}^2 \sum_{k=1}^{d} \|\Psi(\xi_k)\|_2^2 + 2c_{\Psi,u}^2 \sum_{k=1}^{d} \|u_k\|_2^2
        }
        = \sqrt{
            2c_{\Psi,\xi}^2 \|\Psi(\Xi)\|_\mathrm{F}^2 + 2c_{\Psi,u}^2 \|U\|_\mathrm{F}^2
        }.
    \end{align}%
    Substituting this estimate into~\eqref{eq:proof-bound-Delta-literature} yields
    \begin{equation}\label{eq:proof-Delta-error}
        \left\|\begin{bmatrix}
            \Delta A & \Delta B_0 & \Delta \tB
        \end{bmatrix}\right\|
        \leq c_{\Delta,u}
    \end{equation}%
    with $c_{\Delta,u}$ defined in~\eqref{eq:bound-r-Delta-constants}.
    Hence, we establish~\eqref{eq:bilinear-surrogate-perturbed-dynamics} by rewriting~\eqref{eq:Koopman-bilinearization-perturbed} using $
        r_\Delta(\xi,u) 
        = \begin{bmatrix} 
            \Delta A & \Delta B_0 & \Delta \tB 
        \end{bmatrix}    
        \begin{bmatrix} 
            \Psi(\xi) \\ u \\ u\otimes \Psi(\xi)
        \end{bmatrix}
    $.
    The output equation~\eqref{eq:bilinear-surrogate-perturbed-output} can be deduced using similar arguments as in the proof of Proposition~\ref{prop:Koopman-bilinearization-exact}. 
    
    Further, the proportional error bound on $r_\Psi$ follows from~\eqref{eq:Koopman-bilinearization-proportional-error-bound} in Assumption~\ref{ass:Koopman-bilinearization-perturbed}. 
    It remains to show the bound~\eqref{eq:bound-r-Delta} on $r_\Delta$.
    To this end, we leverage the estimate~\eqref{eq:proof-Delta-error} to obtain
    \begin{align}
        \|r_\Delta(\xi,u)\|^2 
        &\leq 
        \left\| 
            \begin{bmatrix} 
                \Delta A & \Delta B_0 & \Delta \tB 
            \end{bmatrix}
        \right\|^2
        \left(
            \|\Psi(\xi)\|^2 + \|u\|^2 + \|u\otimes \Psi(\xi)\|^2
        \right)
        \nonumber\\
        &\leq 
        c_{\Delta,u}^2
        \left(
            \|\Psi(\xi)\|^2 + \|u\|^2 + \|u\otimes \Psi(\xi)\|^2
        \right).
    \end{align}%
    Finally, we use
    \begin{equation}
        \|u\otimes \Psi(\xi)\|
        = \|u\| \|\Psi(\xi)\|
        \leq \max_{u\in\bbU}\|u\| \|\Psi(\xi)\|.
    \end{equation}%
    to establish the error bound~\eqref{eq:bound-r-Delta} with the constants $c_{\Delta,\xi}$ and $c_{\Delta,u}$ defined in~\eqref{eq:bound-r-Delta-constants}.
\end{proof}
Theorem~\ref{thm:Koopman-bilinearization-perturbed} establishes a suitable bilinear input-output characterization of the underlying nonlinear system for a subsequent robust controller design.
As already mentioned earlier, designing a lifted-state-feedback controller for the Koopman surrogate~\eqref{eq:Koopman-bilinearization-perturbed} and, thus, for the input-output representation~\eqref{eq:IO-representation-nonlinear} with extended state $\xi$, corresponds to an \emph{output}-feedback controller for the underlying nonlinear system~\eqref{eq:dynamics-nonlinear}.

The proof of Theorem~\ref{thm:Koopman-bilinearization-perturbed} builds on the same least-squares estimation technique used in~\citet{schmitz:bold:philipp:rosenfelder:eberhard:ebel:worthmann:2025} for local affine-linear regression, where the estimation error is controlled by bounding the residual and the smallest singular value of the data matrix separately.
Two additional layers of work arise from the Koopman setting.
First, the residual is not an externally bounded noise but a structural approximation error, whose norm must be estimated using the proportional bound~\eqref{eq:Koopman-bilinearization-proportional-error-bound}. 
Second, the resulting matrix-level mismatch between true and estimated system matrices must be propagated to a pointwise residual bound, which requires the bilinear structure of the surrogate model.
Consequently, the constants $c_{\Delta,\xi}$ and $c_{\Delta,u}$ in~\eqref{eq:bound-r-Delta-constants} explicitly reflect both the quality of the available data, through $\sigma_{\min}$, and the approximation quality of the lifting function $\Psi$, through the proportionality constants $c_{\Psi,\xi}$ and $c_{\Psi,u}$ in~\eqref{eq:Koopman-bilinearization-proportional-error-bound}.
\begin{remark}
    If the full state is measurable, i.e., $h(x,u)=x$, Theorem~\ref{thm:Koopman-bilinearization-perturbed} provides an alternative characterization of the \emph{learning} error of the Koopman operator approximation from input-state data.
    In particular, in this case, no extended state is needed, but we can replace $\xi$ by $x$ to obtain
    \begin{subequations}\label{eq:state-data-case-surrogate}
        \begin{align}
            \Psi(x_{k+1}) &= A \Psi(x_k) + B_0 u_k + \tB (u_k \otimes \Psi(x_k)) + r_\Psi(x,u) + r_\Delta(x,u),
            \\
            \|r_\Psi(x,u)\| &\leq c_{\Psi,x} \|\Psi(x)\| + c_{\Psi,u} \|u\|, 
            \label{eq:state-data-case-residual-Psi}
            \\
            \|r_\Delta(x,u)\|^2 &\leq c_{\Delta,x}^2 \|\Psi(x)\|^2 + c_{\Delta,u}^2 \|u\|^2
        \end{align}
    \end{subequations}
    for all $x\in\bbX$ and $u\in\bbU$.
    Here, similar to Assumption~\ref{ass:Koopman-bilinearization-perturbed}, we assume the existence of a finite-dimensional perturbed bilinear representation of the form 
    \begin{equation}
        \Psi(x_{k+1}) = A_\mathrm{tr}(x_k) + B_{0,\mathrm{tr}} u_k + \tB_\mathrm{tr} (u_k \otimes \Psi(x_k)) + r_\Psi(x_k,u_k),
    \end{equation}%
    where $r_\Psi$ is bounded by~\eqref{eq:state-data-case-residual-Psi} for some $c_{\Psi,x},c_{\Psi,u}\geq 0$, and compute $A,B_0,\tB$ based on the linear regression problem 
    \begin{equation}
        \min_{A,B_0,\tB} 
        \left\|
            \Psi(X^+) - A \Psi(X) - B_0 U - \tB U_X 
        \right\|_\mathrm{F}
    \end{equation}%
    with the data matrices $X$, $X^+$, $U$, $U_X$ defined analogously to the matrices in~\eqref{eq:data-matrices-IO-regression} and 
    \begin{equation}
        c_{\Delta,x} 
        = c_{\Delta,u} \sqrt{1 + \max_{u\in\bbU}\|u\|^2}
        ,\qquad
        c_{\Delta,u} 
        = \frac{
            \sqrt{
                2c_{\Psi,x}^2 \|\Psi(X)\|_\mathrm{F}^2 + 2c_{\Psi,u}^2 \|U\|_\mathrm{F}^2
            }
        }{
            \sigma_{\min}\left(
                \begin{bmatrix}
                    \Psi(X)\\U\\U_X
                \end{bmatrix}
            \right)
        }.
    \end{equation}%
    In contrast, SafEDMD~\citep{strasser:schaller:worthmann:berberich:allgower:2025,strasser:schaller:worthmann:berberich:allgower:2026} derive proportional error bounds for bilinear approximations of the controlled Koopman generator and controlled Koopman operator by combining multiple autonomous variants and building on the respective error bounds in~\citet[Theorem~3]{schaller:worthmann:philipp:peitz:nuske:2023} and~\citet[Theorem~4]{nuske:peitz:philipp:schaller:worthmann:2023}, respectively.
    This, however, requires multiple data sets collected under specific \emph{constant} control inputs with i.i.d.\ samples, a sampling strategy that may be restrictive in practice. 
    Instead, Theorem~\ref{thm:Koopman-bilinearization-perturbed} and, more specifically, the state-data surrogate~\eqref{eq:state-data-case-surrogate} characterize the learning error for \emph{any} data trajectory, without restricting to particular input choices and data requirements, relying instead on results from noisy least-squares optimization.
    Thus, Theorem~\ref{thm:Koopman-bilinearization-perturbed} is of independent interest beyond the input-output case, and may allow for more flexible sampling strategies, a broader class of control problems to be addressed within the Koopman framework, and potentially more interpretable error bounds due to its direct reliance on noisy least-squares optimization.
\end{remark}
\begin{remark}
    If the input-output data is subject to noise, i.e., the nonlinear input-output dynamics~\eqref{eq:IO-representation-nonlinear} satisfy 
    \begin{equation}
        \xi_{k+1} = f_\xi(\xi_k,u_k) + w_k
    \end{equation}%
    for $\|w_k\|\leq \bar{w}$, the estimate can be straightforwardly adapted. 
    In particular, the noise enters the lifted data matrix $\Psi(\Xi^+)$ and therefore deteriorates the least-squares estimate.
    Then, the mismatch between true values and the least-squares solution is bounded by
    \begin{equation}
        \left\|\begin{bmatrix}
            \Delta A & \Delta B_0 & \Delta \tB
        \end{bmatrix}\right\|
        \leq \frac{
            \sqrt{
                2c_{\Psi,\xi}^2 \|\Psi(\Xi)\|_\mathrm{F}^2 + 2c_{\Psi,u}^2 \|U\|_\mathrm{F}^2
            }
            + \sqrt{d} L_\Psi \bar{w}
        }{
            \sigma_{\min}\left(
                \begin{bmatrix}
                    \Psi(\Xi)\\U\\U_\Xi
                \end{bmatrix}
            \right)
        },
    \end{equation}%
    where we assume Lipschitz continuity of the lifting function $\Psi$.
    Thus, the structure of the residual bound on $r_\Delta(\xi,u)$ remains unchanged.
\end{remark}

\begin{remark}
    In this section, we establish a Koopman-based bilinear \emph{input-output} surrogate representation with guaranteed \emph{proportional} error bounds. 
    This is particularly desirable as it allows for the application of Koopman-based controller designs for bilinear systems with closed-loop guarantees (see~\citet{strasser:worthmann:mezic:berberich:schaller:allgower:2026} and the references therein).
    Here, we exploit the (nonlinear) input-output representation~\eqref{eq:IO-representation-nonlinear} of the unknown nonlinear system~\eqref{eq:dynamics-nonlinear}, relying on uniform observability of the underlying system.
    A possible alternative approach would be to first Koopman bilinearize the original nonlinear system, i.e., employing a state-based lifting function to obtain a bilinear surrogate in the lifted Koopman space, again with proportionally bounded residual error.
    Since this bilinear representation is unknown and the (state-based) lifting function is not accessible, we could use again input-output data to construct a representation with the same input-output behavior as the bilinear surrogate using an extended state.
    Comparing both approaches, i.e., 
    1) input-output characterization for nonlinear systems, then Koopman bilinearization, or
    2) Koopman bilinearization, then input-output characterization for bilinear systems, 
    is an interesting direction for future research.
\end{remark}

\subsection{Koopman-based output-feedback controller design}\label{sec:IO-Koopman-controller-design}
In the following, we use the Koopman-based input-output surrogate established in Theorem~\ref{thm:Koopman-bilinearization-perturbed} to design a robust output-feedback controller with closed-loop guarantees for the nonlinear system~\eqref{eq:dynamics-nonlinear}.
To this end, we build on the state-feedback controller designs established in~\citet{strasser:schaller:worthmann:berberich:allgower:2026} using linear matrix inequalities and in~\citet{strasser:berberich:allgower:2025} using SOS optimization.
Then, we use the proposed extended-state representation~\eqref{eq:IO-representation-nonlinear} of the underlying nonlinear system~\eqref{eq:dynamics-nonlinear} to define an output-feedback control law.
While the generalization from state-feedback to output-feedback design can be done for any robust controller design for bilinear Koopman surrogates with (proportionally) bounded residual error, we rely on the SOS-based design proposed in~\citet[Corollary 4]{strasser:berberich:allgower:2025} as it provides the least conservative (guaranteed) closed-loop properties available in the literature. 

Before stating our main theorem, we introduce some necessary SOS notation.
We denote the set of all polynomials $s$ in the variable $z\in\bbR^N$ with degree $n_d$ and real coefficients by $\bbR[z,n_d]$.
Similarly, we write $\bbR[z,n_d]^{p\times q}$ for the set of all $p \times q$-matrices with elements in $\bbR[z,n_d]$.
We call $S\in\bbR[z,2n_d]^{p\times p}$ an SOS matrix in $x$, denoted by $S\in\mathrm{SOS}[z,2n_d]^p$, if it can be decomposed as $S=T^\top T$ for some $T\in\bbR[z,n_d]^{q\times p}$.
If $(S-\varepsilon I_p)\in\mathrm{SOS}[z,2n_d]^p$ for some $\varepsilon>0$, we say $S$ is strictly SOS and write $S\in\mathrm{SOS}_+[z,2n_d]^p$.
\begin{theorem}[Exponentially stabilizing output-feedback controller design]\label{thm:Koopman-control-output-feedback}
    Suppose Assumption~\ref{ass:Koopman-bilinearization-perturbed}, ensuring an approximate bilinear Koopman representation, holds.
    Let $\alpha\in\bbN$, $\beta\in\bbN_0$ with $\alpha\geq \beta$, $P=P^\top\succ 0$ of size $N\times N$, $L_\mathrm{n}\in\bbR[z,2\alpha-1]^{m\times N}$, $\tau_1,\tau_2\in\mathrm{SOS}_+[z,2\beta]$, $u_\mathrm{d}\in\mathrm{SOS}_+[z,2\alpha]$, and $\rho > 0$ such that
    \begin{equation}\label{eq:Koopman-control-output-feedback}
        \begin{bmatrix}
            q_{11}(z)
            & 0 
            & 0 
            & 0 
            & 0
            & q_{15}(z)
            \\
            0 
            & \frac{\tau_1(z)}{2 c_{\Psi,\xi}^2} I_N
            & 0 
            & 0 
            & 0 
            & u_\mathrm{d}(z) P
            \\
            0 
            & 0 
            & \frac{\tau_1(z)}{2 c_{\Psi,u}^2} I_m
            & 0 
            & 0 
            & L_\mathrm{n}(z) 
            \\
            0 
            & 0
            & 0 
            & \frac{\tau_2(z)}{c_{\Delta,\xi}^2} I_N 
            & 0
            & u_\mathrm{d}(z) P 
            \\
            0 
            & 0 
            & 0 
            & 0 
            & \frac{\tau_2(z)}{c_{\Delta,u}^2} I_m
            & L_\mathrm{n}(z)
            \\
            q_{15}(z)^\top
            & u_\mathrm{d}(z) P
            & L_\mathrm{n}(z)^\top 
            & u_\mathrm{d}(z) P 
            & L_\mathrm{n}(z)^\top
            & u_\mathrm{d}(z) (P - \rho I_N)
        \end{bmatrix}
        \in\mathrm{SOS}[z,2\alpha]^{4N+2n_u}
    \end{equation}%
    with 
    \begin{equation}
        q_{11}(z) = u_\mathrm{d}(z) P - (\tau_1(z) + \tau_2(z)) I_N,
        \qquad
        q_{15}(z) = u_\mathrm{d}(z)AP + B_0 L_\mathrm{n}(z) + \tB(L_\mathrm{n}(z)\otimes z),
    \end{equation}%
    and $z\in\bbR^N$ holds.
    Then, the output-feedback control law 
    \begin{equation}\label{eq:control-law-output-feedback}
        \mu(\xi) = \frac{1}{u_\mathrm{d}(\Psi(\xi))} L_\mathrm{n}(\Psi(\xi)) P^{-1} \Psi(\xi)
    \end{equation}%
    achieves \emph{robust exponential stability} of the perturbed bilinear Koopman surrogate~\eqref{eq:bilinear-surrogate-perturbed-dynamics} and, thus, \emph{exponential stability} of the nonlinear extended-state dynamics~\eqref{eq:IO-representation-nonlinear-state} for all initial conditions $\hat{\xi}\in\Omega(c^*)$ with
    \begin{subequations}\label{eq:Omega-c-star}
        \begin{align}
            \Omega(c) &= \big\{
                \xi\in\bbR^{L(m+p)}
                \mid
                \Psi(\xi)^\top P^{-1} \Psi(\xi) \leq c
            \big\}
            \\
            c^* &= \argmax\big\{
                c\in\bbR_+
                \mid 
                \Omega(c)\subseteq \bbE \text{ and } \mu(\xi)\in\bbU \text{ for all } \xi\in\Omega(c)
            \big\}.
        \end{align}
    \end{subequations}
\end{theorem}%
\begin{proof}
    The proof follows the arguments in~\citet[Corollary 4]{strasser:berberich:allgower:2025}.
    However, we generalize the result to two different residual error bounds on $r_\Psi$ and $r_\Delta$, which are required in the present paper due to the surrogate structure in~\eqref{eq:bilinear-surrogate-perturbed}.
    In particular, we account for an additional uncertainty channel in the robust controller design.
    In the following, we provide the key differences relative to the proof in~\citet[Corollary 4]{strasser:berberich:allgower:2025}.
    
    We introduce $K_\mathrm{n}(z) = L_\mathrm{n}(z) P^{-1}$ and 
    \begin{equation}
        \cA(z) = u_\mathrm{d}(z) A + B_0 K_\mathrm{n}(z) + \tB(K_\mathrm{n}(z)\otimes z),
    \end{equation}%
    where we exploit $L_\mathrm{n}(z)\otimes z = (K_\mathrm{n}(z)\otimes z)P$.
    Then, substituting $K_\mathrm{n}(z)$ and $\cA(z)$ into~\eqref{eq:Koopman-control-output-feedback}, applying the Schur complement twice, and permuting rows and columns~\citep[cf.][]{strasser:berberich:allgower:2025} yields
    \begin{equation}
        \mathfrak{B}^\top 
        \left[\begin{array}{cc|ccc|ccc|c}
            -\frac{1}{u_\mathrm{d}(z)} P & 0
            & 0 & 0 & 0
            & 0 & 0 & 0
            & 0
            \\
            0 & u_\mathrm{d}(z) P
            & 0 & 0 & 0 
            & 0 & 0 & 0 
            & 0
            \\\hline
            0 & 0
            & \frac{\tau_1(z)}{2 c_{\Psi,\xi}^2} I_N & 0 & 0 
            & 0 & 0 & 0
            & 0
            \\
            0 & 0
            & 0 & \frac{\tau_1(z)}{2 c_{\Psi,u}^2} I_m & 0 
            & 0 & 0 & 0
            & 0
            \\
            0 & 0
            & 0 & 0 & -\tau(z) I_N
            & 0 & 0 & 0
            & 0
            \\\hline
            0 & 0
            & 0 & 0 & 0
            & \frac{\tau_2(z)}{c_{\Delta,\xi}^2} I_N & 0 & 0 
            & 0
            \\
            0 & 0
            & 0 & 0 & 0
            & 0 & \frac{\tau_2(z)}{2 c_{\Delta,u}^2} I_m & 0 
            & 0
            \\
            0 & 0
            & 0 & 0 & 0
            & 0 & 0 & -\tau(z) I_N
            & 0
            \\\hline
            0 & 0
            & 0 & 0 & 0
            & 0 & 0 & 0
            & \frac{u_\mathrm{d}(z)}{\rho} P^{2}
        \end{array}\right]
        \mathfrak{B}
        \succeq 0
    \end{equation}%
    for all $z\in\bbR^N$ with
    \begin{equation}
        \mathfrak{B}^\top = \left[\begin{array}{cc|ccc|ccc|c}
            \cA(z) & -I_N
            & 0 & 0 & I_N
            & 0 & 0 & I_N
            & 0
            \\
            u_\mathrm{d}(z) I_N & 0
            & -I_N & 0 & 0
            & 0 & 0 & 0
            & 0
            \\            
            K_\mathrm{n}(z) & 0
            & 0 & -I_m & 0
            & 0 & 0 & 0 
            & 0
            \\
            u_\mathrm{d}(z) I_N & 0
            & 0 & 0 & 0 
            & -I_N & 0 & 0 
            & 0
            \\
            K_\mathrm{n}(z) & 0
            & 0 & 0 & 0
            & 0 & -I_m & 0
            & 0
            \\
            u_\mathrm{d}(z) I_N & 0 
            & 0 & 0 & 0 
            & 0 & 0 & 0
            & -I_N
        \end{array}\right].
    \end{equation}%
    Applying the dualization lemma~\citep[Lemma 4.9]{scherer:weiland:2000} yields
    \begin{equation}\label{eq:proof-dualized-QMI}
        \resizebox{0.985\textwidth}{!}{$\displaystyle 
            \tilde{\mathfrak{B}}^\top 
            \left[\begin{array}{cc|ccc|ccc|c}
                -u_\mathrm{d}(z) P^{-1} & 0
                & 0 & 0 & 0
                & 0 & 0 & 0
                & 0
                \\
                0 & \frac{1}{u_\mathrm{d}(z)}P^{-1}
                & 0 & 0 & 0 
                & 0 & 0 & 0 
                & 0
                \\\hline
                0 & 0
                & \frac{2 c_{\Psi,\xi}^2}{\tau_1(z)} I_N & 0 & 0 
                & 0 & 0 & 0
                & 0
                \\
                0 & 0
                & 0 & \frac{2 c_{\Psi,u}^2}{\tau_1(z)} I_m & 0 
                & 0 & 0 & 0
                & 0
                \\
                0 & 0
                & 0 & 0 & -\frac{1}{\tau(z)} I_N
                & 0 & 0 & 0
                & 0
                \\\hline
                0 & 0
                & 0 & 0 & 0
                & \frac{c_{\Delta,\xi}^2}{\tau_2(z)} I_N & 0 & 0 
                & 0
                \\
                0 & 0
                & 0 & 0 & 0
                & 0 & \frac{2 c_{\Delta,u}^2}{\tau_2(z)} I_m & 0 
                & 0
                \\
                0 & 0
                & 0 & 0 & 0
                & 0 & 0 & -\frac{1}{\tau(z)} I_N
                & 0
                \\\hline
                0 & 0
                & 0 & 0 & 0
                & 0 & 0 & 0
                & \frac{\rho}{u_\mathrm{d}(z)} P^{-2}
            \end{array}\right]
            \tilde{\mathfrak{B}}
            \preceq 0
        $}
    \end{equation}%
    for all $z\in\bbR^N$ with
    \begin{equation}
        \tilde{\mathfrak{B}}^\top = \left[\begin{array}{cc|ccc|ccc|c}
            I_N & \cA(z)^\top
            & u_\mathrm{d}(z) I_N & K_\mathrm{n}(z)^\top & 0
            & u_\mathrm{d}(z) I_N & K_\mathrm{n}(z)^\top & 0
            & u_\mathrm{d}(z) I_N
            \\
            0 & I_N 
            & 0 & 0 & I_N
            & 0 & 0 & I_N
            & 0
        \end{array}\right],
    \end{equation}%
    where the outer matrix $\mathfrak{B}$ follows according to the discussion in~\citet[Section 8.1.2]{scherer:weiland:2000}
    Further, we observe that the proportional error bounds~\eqref{eq:Koopman-bilinearization-proportional-error-bound} and~\eqref{eq:bound-r-Delta} on the residuals $r_\Psi$ and $r_\Delta$ imply the quadratic matrix inequalities 
    \begin{equation}
        \begin{bmatrix}
            \star
        \end{bmatrix}^\top 
        \begin{bmatrix}
            c_{\Delta,\xi}^2 I_N & 0 & 0 \\
            0 & c_{\Delta,u}^2 I_m & 0 \\
            0 & 0 & -I_N
        \end{bmatrix}
        \begin{bmatrix}
            \Psi(\xi) \\ u \\ r_\Delta(\xi,u)
        \end{bmatrix}
        \geq 0,
        \quad
        \begin{bmatrix}
            \star
        \end{bmatrix}^\top
        \begin{bmatrix}
            2c_{\Psi,\xi}^2 I_N & 0 & 0 \\
            0 & 2c_{\Psi,u}^2 I_m & 0 \\
            0 & 0 & -I_N
        \end{bmatrix} 
        \begin{bmatrix}
            \Psi(\xi) \\ u \\ r_\Psi(\xi,u)
        \end{bmatrix}
        \geq 0,
    \end{equation}%
    respectively.
    Then, exploiting the generalized S-procedure~\citep[Lemma 2.1]{tan:2006},~\eqref{eq:proof-dualized-QMI} implies that the robust controller $\mu(\xi)$ exponentially stabilizes the perturbed bilinear surrogate~\eqref{eq:bilinear-surrogate-perturbed-dynamics} and, thus, with Assumption~\ref{ass:Koopman-bilinearization-perturbed}, also the underlying nonlinear extended-state dynamics~\eqref{eq:IO-representation-nonlinear-state}~\citep[compare][]{strasser:berberich:allgower:2025}.
\end{proof}
Theorem~\ref{thm:Koopman-control-output-feedback} establishes, to the best of the authors' knowledge, the first Koopman-based output-feedback controller design method with closed-loop guarantees for the underlying nonlinear system.
In particular, we establish exponential stability of the extended-state system~\eqref{eq:IO-representation-nonlinear}, which has the equivalent input-output behavior as the underlying nonlinear system (compare Proposition~\ref{prop:IO-representation-nonlinear} under Assumptions~\ref{ass:smoothness} and~\ref{ass:lower-bound-Jacobian-observability-map}).
Hence, Theorem~\ref{thm:Koopman-control-output-feedback} guarantees closed-loop guarantees for the nonlinear system~\eqref{eq:dynamics-nonlinear} from input-output data.
The parametrization used in Theorem~\ref{thm:Koopman-bilinearization-perturbed} for the Koopman-based surrogate can be obtained using the collected input-output data $\cD$ and requires no state information.
This paves the way forward to Koopman-based control in practical applications, where the full state is typically not accessible.

Note that the output-feedback controller established in Theorem~\ref{thm:Koopman-control-output-feedback} is based on the SOS-based \emph{state}-feedback design proposed in~\citet{strasser:berberich:allgower:2025}.
However, the Koopman-based bilinear surrogate~\eqref{eq:bilinear-surrogate-perturbed} derived for the input-output representation~\eqref{eq:IO-representation-nonlinear} is applicable more generally, i.e., it can be combined with \emph{any} robust state-feedback controller design for (Koopman-based) bilinear surrogate models that provides closed-loop guarantees for the underlying nonlinear state-space system.
In particular, the key contribution of Theorem~\ref{thm:Koopman-control-output-feedback} lies in its explicit connection between the input-output behavior of the nonlinear system~\eqref{eq:dynamics-nonlinear} and the extended-state representation~\eqref{eq:IO-representation-nonlinear}, as formalized in Theorem~\ref{thm:Koopman-bilinearization-perturbed}.
This connection enables an \emph{output}-feedback controller design.
Since this relationship is already established in Theorem~\ref{thm:Koopman-bilinearization-perturbed}, the resulting output-feedback construction directly extends to alternative state-feedback designs based on bilinear surrogate models with (proportionally) bounded residuals.

The following corollary relates the established exponential stability of the \emph{extended} state $\xi\in\bbR^{L(m+n)}$ to the original state $x\in\bbR^n$.
\begin{corollary}[Exponential convergence of original state]\label{cor:Koopman-control-output-feedback-convergence}
    Suppose Assumptions~\ref{ass:smoothness},~\ref{ass:lower-bound-Jacobian-observability-map}, and~\ref{ass:Koopman-bilinearization-perturbed} hold.
    Let $\alpha\in\bbN$, $\beta\in\bbN_0$ with $\alpha\geq \beta$, $P=P^\top\succ 0$ of size $N\times N$, $L_\mathrm{n}\in\bbR[z,2\alpha-1]^{m\times N}$, $\tau_1,\tau_2\in\mathrm{SOS}_+[z,2\beta]$, $u_\mathrm{d}\in\mathrm{SOS}_+[z,2\alpha]$, and $\rho > 0$ such that~\eqref{eq:Koopman-control-output-feedback}
    with $z\in\bbR^N$ holds.
    Then, the output-feedback control law~\eqref{eq:control-law-output-feedback} based on the extended state $\xi$ in~\eqref{eq:extended-state} achieves \emph{exponential convergence} of the state $x$ of the nonlinear system~\eqref{eq:dynamics-nonlinear} to the origin for all initial conditions $\hat{\xi}\in\Omega(c^*)$, where $\Omega(c)$ and $c^*$ are as defined in~\eqref{eq:Omega-c-star}.
\end{corollary}
\begin{proof}
    Theorem~\ref{thm:Koopman-control-output-feedback} already establishes exponential stability of the extended-state representation~\eqref{eq:IO-representation-nonlinear} using Assumption~\ref{ass:Koopman-bilinearization-perturbed}.
    More precisely, exponential stability of the extended-state system ensures the existence of constants $c>0$ and $\lambda\in(0,1)$ such that $\|\xi_{k}\| \leq c \lambda^{k} \|\xi_0\|$ for all $k\geq 0$.
    Using Proposition~\ref{prop:IO-representation-nonlinear} and uniform observability based on Assumptions~\ref{ass:smoothness} and~\ref{ass:lower-bound-Jacobian-observability-map}, there exists the well-defined maps $\cO_L$ and $\cO_L^{-1}$ with $y_{k-L,k-1} = \cO(x_{k-L},\useq_{k-L,k-1})$ and $x_{k-L} = \cO_L^{-1}(\xi_k)$, where $\cO$ satisfies the bi-Lipschitz condition in~\eqref{eq:bi-Lipschitz-cO}.
    We choose $\useq = \useq_{k-L:k-1}$ and set $x_1=x_{k-L}$ and $x_2=0$ and obtain
    \begin{align}
        \alpha \|x_{k-L}\|
        &\leq \|\cO_L(x_{k-L},\useq_{k-L:k-1}) - \cO_L(0,\useq_{k-L:k-1})\|
        = \|y_{k-L:k-1} - \cO_L(0,\useq_{k-L:k-1})\|
        \\
        &\leq \|y_{k-L:k-1}\| + \cO_L(0,\useq_{k-L:k-1})\|.
    \end{align}
    Since $y_{k-L:k-1}$ is part of $\xi_k$, we know $\|y_{k-L:k-1}\|\leq\|\xi_k\|$.
    Using Assumption~\ref{ass:smoothness}, $\cO$ is $C^1$ in $\useq\in\bbU^L$ and $\bbU$ is compact, i.e., there exists $L_u>0$ such that 
    \begin{equation}
        \| \cO_L(0,\useq_{k-L:k-1}) - \cO_L(0,\mathbf{0}) \| \leq L_u \|\useq_{k-L,k-1}\|.
    \end{equation}%
    Note that $\cO_L(0,\mathbf{0})=0$ and $\useq_{k-L:k-1}$ is part of $\xi_k$.
    Then, we deduce  
    \begin{equation}
        \| \cO_L(0,\useq_{k-L:k-1}) \| 
        \leq L_u \|\useq_{k-L,k-1}\|
        \leq L_u \|\xi_k\|.
    \end{equation}%
    Hence, 
    \begin{equation}
        \|x_{k-L}\|
        \leq \frac{1+L_u}{\alpha} \|\xi_k\|
        \leq \frac{1+L_u}{\alpha} \lambda^k \|\xi_0\|.
    \end{equation}%
    Thus, we conclude convergence of the unknown state $x$ to the origin, i.e., $x_k\to 0$ for $k\to \infty$.    
\end{proof}
Corollary~\ref{cor:Koopman-control-output-feedback-convergence} links the closed-loop properties for the extended-state representation~\eqref{eq:IO-representation-nonlinear} to the original nonlinear state-space system~\eqref{eq:dynamics-nonlinear}.
In particular, by using the proposed output-feedback controller, we ensure exponential convergence of the state $x$ to the origin.

\begin{remark}\label{rem:SOS-combined-residuals}
    Given the possibly high-dimensional lifting dimension $N$, the SOS program~\eqref{eq:Koopman-control-output-feedback} may be computationally challenging to solve.
    While we leave a structured analysis of possible model order reduction techniques for the proposed output-feedback controller design to future work, a direct dimensionality reduction follows from combining the uncertainty characterizations of both residual terms $r_\Psi$ and $r_\Delta$.
    More precisely, solving
    \begin{equation}\label{eq:Koopman-control-output-feedback-combined-residuals}
        \begin{bmatrix}
            q_{11}(z)
            & 0 
            & 0 
            & q_{15}(z)
            \\
            0 
            & \frac{\tau_1(z)}{2 c_{\Psi,\xi}^2 + c_{\Delta,\xi}^2} I_N
            & 0 
            & u_\mathrm{d}(z) P
            \\
            0 
            & 0 
            & \frac{\tau_1(z)}{2 c_{\Psi,u}^2 + c_{\Delta,u}^2} I_m
            & L_\mathrm{n}(z)
            \\
            q_{15}(z)^\top
            & u_\mathrm{d}(z) P
            & L_\mathrm{n}(z)^\top
            & u_\mathrm{d}(z) (P - \rho I_N)
        \end{bmatrix}
        \in\mathrm{SOS}[z,2\alpha]^{3N+n_u}
    \end{equation}%
    instead of~\eqref{eq:Koopman-control-output-feedback} in order to obtain the control law $\mu$ in~\eqref{eq:control-law-output-feedback} reduces the dimension of the SOS program by introducing additional conservatism through the combination of the two residual error bounds.
\end{remark}

\section{Numerical example}\label{sec:numerical-examples}
In this section, we illustrate our theoretical findings in numerical simulations.
All our simulations are conducted on an i7 notebook using Yalmip~\citep{lofberg:2004} with the semi-definite programming solver MOSEK~\citep{mosek:2026} in MATLAB R2026a.

We consider the continuous-time nonlinear dynamical system
\begin{subequations}
    \begin{align}
        \dot{x}_1(t) &= -2 x_1(t), \\
        \dot{x}_2(t) &= -x_1(t)^2 + x_2(t) + u, \\
        y(t) &= x_1(t) + x_2(t),
    \end{align}
\end{subequations}
which is discretized with a 4th order Runge-Kutta method~\citep{dormand:prince:1980} for the sampling time $T_s = 0.1$.
Let $\bbX = [-1.5,1.5]^2$, $\bbU = [-3,3]$, and $\bbY=[-3,3]$.
We collect $d$ data trajectories of length $L$, where we uniformly sample the initial state from $\bbX$ and evaluate the dynamics to obtain the corresponding input-output measurements.
This allows us to construct the extended-state data as in~\eqref{eq:data-extended-state}.

In the following, we illustrate the effectiveness of the proposed Koopman-based surrogate model for prediction (Section~\ref{sec:numerical-examples:prediction}) as well as for an output-feedback controller design (Section~\ref{sec:numerical-examples:control}).

\subsection{Open-loop prediction}\label{sec:numerical-examples:prediction}
First, we investigate the prediction capability of the proposed Koopman-based bilinear surrogate for the input-output behavior of the underlying nonlinear system.
To this end, we employ a polynomial lifting function $\Psi(\xi)$, which contains all monomials up to degree $n_d$. 
Then, we choose $d=100$ and vary the delay length $L$ as well as the degree $n_d$ to build the surrogate in~\eqref{eq:bilinear-surrogate-perturbed} purely based on input-output data.
We consider the nominal bilinear surrogate 
\begin{equation}\label{eq:Koopman-surrogate-nominal}
    \check{\Psi}_{k+1} = A \check{\Psi}_k + B_0 u_k + \tB (u_k \otimes \check{\Psi}_k), 
    \qquad
    \check{\Psi}_0 = \Psi(\xi_0)
\end{equation}%
and compare the open-loop prediction error 
\begin{equation}
    \frac{1}{N} \| \Psi(\xi_k) - \check{\Psi}_k \|
\end{equation}%
for random inputs $u_k$ uniformly sampled from $\bbU$.
The resulting open-loop prediction error, averaged over 50 runs, is depicted in Fig.~\ref{fig:numerics-prediction-KoopmanIO}.
Here, we illustrate different combinations of $L$ and $n_d$.
\begin{figure}[t]
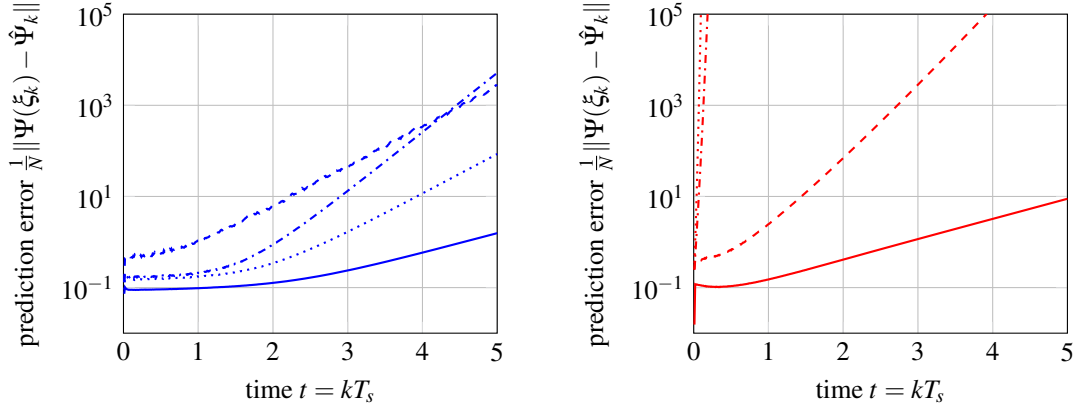

    \centering
    \setbox1=\hbox{\begin{tikzpicture}[baseline]%
        \draw[black, thick] (0,.6ex)--++(1.25em,0);
    \end{tikzpicture}}%
    \setbox2=\hbox{\begin{tikzpicture}[baseline]%
        \draw[black, thick, dashed] (0,.6ex)--++(1.25em,0);
    \end{tikzpicture}}%
    \setbox3=\hbox{\begin{tikzpicture}[baseline]%
        \draw[black, thick, dotted] (0,.6ex)--++(1.25em,0);
    \end{tikzpicture}}%
    \setbox4=\hbox{\begin{tikzpicture}[baseline]%
        \draw[black, thick, dashdotted] (0,.6ex)--++(1.25em,0);
    \end{tikzpicture}}%
    \subfloat[Proposed bilinear surrogate~\eqref{eq:Koopman-bilinearization-perturbed}.]{\label{fig:numerics-prediction-KoopmanIO}
        \input{Fig/prediction-benchmark-KoopmanIO}
    }
    \hfill
    \subfloat[Linear EDMDc surrogate~\eqref{eq:linear-EDMDc-surrogate}.]{\label{fig:numerics-prediction-EDMDc}
        \input{Fig/prediction-benchmark-EDMDc}
    }
    \caption{Open-loop prediction error of the (nominal) Koopman-based bilinear surrogate~\eqref{eq:Koopman-surrogate-nominal} compared to a \emph{linear} EDMDc surrogate for $n_d=1$, $L=100$~(\usebox1), $n_d=2$, $L=3$~(\usebox2), $n_d=2$, $L=10$~(\usebox3), and $n_d=3$, $L=6$~(\usebox4). The resulting errors are averaged over 50 independent runs.}
    \label{fig:numerics-prediction}
\end{figure}
We emphasize that choosing $\Psi(\xi)=\xi$, i.e., $n_d=1$, yields a satisfactory prediction error for a sufficiently large delay length $L$.
Further, as we see for $n_d=2$, increasing $L$ improves the prediction capabilities of the proposed bilinear surrogate.
Here, the initial prediction error at $t=0$ indicates the approximation quality of the employed delay embedding of depth $L$ together with the $N$-dimensional nonlinear lifting function $\Psi$.

We compare the prediction error to a \emph{linear} EDMDc-based Koopman surrogate as in~\eqref{eq:EDMDc-linear-dynamics}, i.e., 
\begin{equation}\label{eq:linear-EDMDc-surrogate}
    \check{\Psi}_{k+1}^\mathrm{EDMDc} = A^\mathrm{EDMDc} \check{\Psi}_k^\mathrm{EDMDc} + B^\mathrm{EDMDc} u_k, 
    \qquad
    \check{\Psi}_0^\mathrm{EDMDc} = \Psi(\xi_0)
\end{equation}%
where we use the same polynomial lifting function $\Psi(\xi)$ and the same collected data. 
Notably, the prediction error of EDMDc is worse than our proposed bilinear surrogate for all combinations of $n_d$ and $L$ studied in this paper.
More precisely, the prediction error of EDMDc may quickly explode and, therefore, EDMDc is sensitive to the employed delay length $L$ and chosen degree $n_d$ of the lifting function.
Thus, our proposed bilinear Koopman surrogate~\eqref{eq:Koopman-bilinearization-perturbed} offers a more robust prediction w.r.t. the employed lifting and input-output characterization of the underlying nonlinear system.

\subsection{Closed-loop simulation}\label{sec:numerical-examples:control}
Now, we apply the proposed output-feedback controller design in Theorem~\ref{thm:Koopman-control-output-feedback} to show that the output of the nonlinear system converges to the origin.
To this end, we choose $d=100$ with $L=1$ and a polynomial lifting function $\Psi(\xi)$ including all monomials up to degree $n_d=2$.
For the bilinearization error in Assumption~\ref{ass:Koopman-bilinearization-perturbed}, we assume the residual error bound~\eqref{eq:Koopman-bilinearization-proportional-error-bound} to hold with constants $c_{\Psi,\xi}=c_{\Psi_u}=1\mathrm{e}-4$.
Then, we compute the constants $c_{\Delta,\xi}$, $c_{\Delta,u}$ for the error bound~\eqref{eq:bound-r-Delta} on the residual $r_\Delta$ according to~\eqref{eq:bound-r-Delta-constants}.

To solve the SOS program~\eqref{eq:Koopman-control-output-feedback}, we choose $\alpha=\beta=1$ and 
\begin{equation}
    u_\mathrm{d}(\xi) 
    = \begin{bmatrix}
        1 \\ \xi_1 \\ \xi_2 \\ \xi_1^2 \\ \xi_1\xi_2 \\ \xi_2^2
    \end{bmatrix}^\top
    \begin{bmatrix}
        1 & 0 & 0 & 0 & 0 & 0 \\
        0 & 1 & 0.5 & 0.5 & 0.5 & 0.5 \\
        0 & 0.5 & 1 & 0.5 & 0.5 & 0.5 \\
        0 & 0.5 & 0.5 & 1 & 0.5 & 0.5 \\
        0 & 0.5 & 0.5 & 0.5 & 1 & 0.5 \\
        0 & 0.5 & 0.5 & 0.5 & 0.5 & 1
    \end{bmatrix}
    \begin{bmatrix}
        1 \\ \xi_1 \\ \xi_2 \\ \xi_1^2 \\ \xi_1\xi_2 \\ \xi_2^2
    \end{bmatrix}
    \in\mathrm{SOS}[\Psi(\xi),2].
\end{equation}%
Then, the output-feedback control law is computed as in~\eqref{eq:control-law-output-feedback}.
To demonstrate its effectiveness, we sample 100 initial conditions $x_{-L}$ uniformly from $\bbX$ and simulate the nonlinear dynamics under a random input sequence $\useq_{-L,-1}$, drawn uniformly from $\bbU^L$, to initialize the extended state as in~\eqref{eq:extended-state-initialization}.
The resulting input-output measurements are used to construct the extended state $\xi_0$, after which the control law~\eqref{eq:control-law-output-feedback} is applied.
Fig.~\ref{fig:numerics-control} shows the closed-loop output trajectories $y$ for all 100 initial conditions, each of which converges to the origin.
\begin{figure}[t]
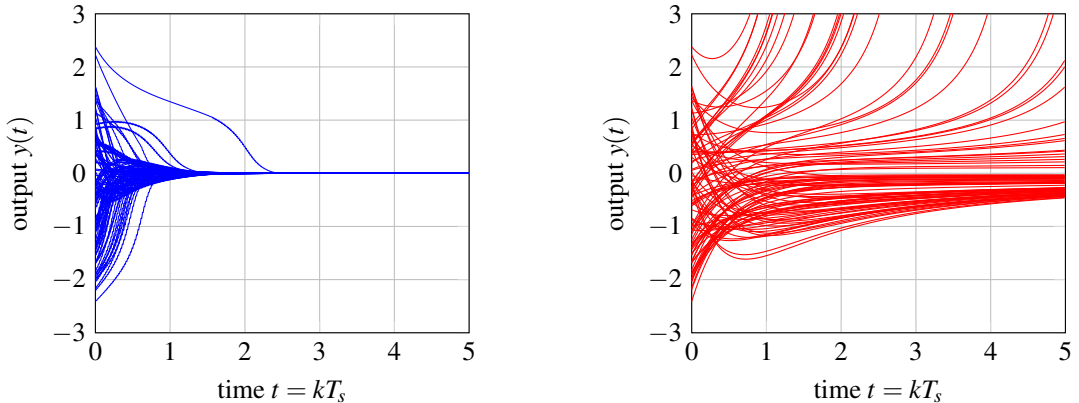

    \centering
    \subfloat[Proposed output-feedback controller~\eqref{eq:control-law-output-feedback}.]{\label{fig:numerics-control-KoopmanIO}
        \input{Fig/control-benchmark-y-KoopmanIO}
    }
    \hfill
    \subfloat[LQR based on EDMDc surrogate~\eqref{eq:linear-EDMDc-surrogate}.]{\label{fig:numerics-control-EDMDc}
        \input{Fig/control-benchmark-y-EDMDc}
    }
    \caption{Closed-loop simulation of the underlying nonlinear system with the proposed output feedback controller $\mu$ compared to an EDMDc-based LQR controller $\mu_\mathrm{LQR}$.}
    \label{fig:numerics-control}
\end{figure}

We compare the resulting closed-loop behavior against one of the most widely used Koopman-based controller design approaches in the literature, namely a linear-quadratic regulator (LQR) designed on the linear EDMDc surrogate~\eqref{eq:linear-EDMDc-surrogate}~\citep{brunton:brunton:proctor:kutz:2016}.
The LQR control law takes the form $\mu_\mathrm{LQR}(\xi) = -K \Psi(\xi)$, where $K_\mathrm{LQR}$ is obtained by solving the corresponding Riccati equation with $Q = R = I$.
Unlike our proposed output-feedback controller~\eqref{eq:control-law-output-feedback}, the LQR controller $\mu_\mathrm{LQR}$ offers no closed-loop stability guarantees and fails to reliably stabilize the system.
By exploiting the bilinear Koopman surrogate of the input-output representation within an SOS framework, we obtain rigorous closed-loop guarantees (Corollary~\ref{cor:Koopman-control-output-feedback-convergence}) and successfully stabilize all initial conditions.

In the simulation above, we choose a delay length of $L = 1$ for the extended state $\xi$ and design the output-feedback controller $\mu$ accordingly.
As observed in the prediction results of Section~\ref{sec:numerical-examples:prediction}, a larger delay length $L$ and richer lifting function $\Psi(\xi)$ would generally be beneficial.
However, the SOS-based controller design is computationally demanding with a computational complexity of $\cO((N^{2\alpha+1})^6)$ and therefore limited in the dimension of the extended state and the nonlinear lifting $\Psi$ that can be practically handled.
An important direction for future work is thus the investigation of model order reduction schemes for the bilinear Koopman surrogate, which would render the controller design feasible for more appropriate choices of $L$ and $\Psi(\xi)$.
Beyond SOS, combining the derived bilinear Koopman surrogate with other controller design techniques, such as MPC, offers another compelling direction to address this limitation.

\section{Conclusion}\label{sec:conclusion}
In this paper, we proposed a data-driven output-feedback controller design method for nonlinear systems that provides provable closed-loop guarantees while relying solely on measured input-output data. 
By combining Koopman operator theory with an extended state representation constructed from input-output trajectories, we derived a bilinear surrogate model directly from data, on which we applied robust state-feedback methods. 
Exploiting the observability of the underlying nonlinear system, we established exponential stability of the extended state and, consequently, exponential convergence of the original system state to the origin, with numerical simulations confirming the theoretical findings. 
To the best of our knowledge, this is the first Koopman-based controller design framework that operates exclusively on input-output data with rigorous closed-loop guarantees, addressing a fundamental gap in the data-driven control literature.
Important directions for future work include extending the framework to handle noisy measurements with finite-sample error quantification or combining the design with model order reduction techniques to enhance practical applicability.

\section*{Funding}
This work was supported by the Deutsche Forschungsgemeinschaft (DFG, German Research Foundation) [AL 316/15-1 -- 468094890 to F.A., 579821331 to J.B.].

\section*{Acknowledgments}
R. Strässer thanks the Graduate Academy of the SC SimTech, Germany for its support.

\bibliographystyle{abbrvnat}
\bibliography{literature}

\appendix
\section{Input-output representation of linear system}
\phantomsection\label{app:IO-representation-linear}
Since the state is not accessible in our setting and only input-output data can be gathered, we need an alternative system characterization of the underlying system. 
To set the stage, we first recall relevant results for linear systems, which serve as a foundation for the later derived input-output representations for bilinear and general nonlinear systems.
More precisely, we leverage ideas from subspace identification to define a non-minimal representation of linear state-space systems using an extended state.
In this section, we consider the discrete-time linear system
\begin{subequations}\label{eq:dynamics-linear}
    \begin{align}
        z_{k+1} &= A z_k + B_0 u_k, \qquad z_0 \in\bbR^N, \\
        y_k &= C z_k + D u_k
    \end{align}%
\end{subequations}%
with $z\in\bbR^N$, $u\in\bbR^m$, and $y\in\bbR^p$, where the state $z_k$ is not available.
Since the minimal internal state $z_k$ is unknown, the core idea is to use past input-output measurements to infer knowledge about the initial condition.
To this end, we characterize the sequence of the past $L>0$ outputs via 
\begin{equation}\label{eq:IO-representation-linear-yseq}
    \yseq_{k-L:k-1}
    = \underbrace{
        \begin{bmatrix}
            C \\ CA \\ \vdots \\ CA^{L-1}
        \end{bmatrix}
    }_{\hat{\cH}_L^{(z_0)}} 
    z_{k-L}
    + \underbrace{
        \begin{bmatrix}
            D & 0 & \cdots & 0 \\
            CB_0 & D & \cdots & 0 \\
            \vdots & \vdots & \ddots & \vdots \\
            CA^{L-2} B_0 & CA^{L-1}B_0 & \cdots & D
        \end{bmatrix}
    }_{\hat{\cH}_L^{(u)}}
    \useq_{k-L:k-1}.
\end{equation}%
\begin{definition}[Lag of linear systems]\label{def:lag-linear}
    The lag $\nu$ of the linear system~\eqref{eq:dynamics-linear} is defined as the smallest integer $L$ such that $ \hat{\cH}_L^{(z_0)} $ has full column rank.
\end{definition}
The lag corresponds to the observability index and determines whether the initial condition $z_{k-L}$ can be uniquely determined from past input-output sequences.
Further, we construct an extended state vector from stacked past inputs and outputs over a finite horizon $L$ as in~\eqref{eq:extended-state}.
Now we can formulate an equivalent representation of the input-output behavior of length $L$ for system~\eqref{eq:dynamics-linear}.
\begin{proposition}[Input-output representation of linear systems]\label{prop:IO-representation-linear}
    Let $L\geq \nu$.
    Then, any trajectory $\{u_k,y_k\}_{k=-L}^{d}$ of~\eqref{eq:dynamics-linear} can be explained by the extended-state system 
    \begin{subequations}\label{eq:IO-representation-linear}
        \begin{align}
            \xi_{k+1} &= \tA \xi_k + \tB u_k, \\
            y_k &= \tC \xi_k + \tD u_k
        \end{align}%
    \end{subequations}%
    with an extended state~\eqref{eq:extended-state} initialized as in~\eqref{eq:extended-state-initialization}, where the matrices $\tA$, $\tB$, $\tC$, and $\tD$ are defined by 
    \begin{subequations}\label{eq:extended-state-matrices-linear}
        \begin{align}
            \label{eq:extended-state-matrices-linear-A}
            \tA &= \left[\begin{array}{cccc|cccc}
                0 & I & \cdots & 0 & 0 & \cdots & 0 & 0 \\
                \vdots & \vdots & \ddots & \vdots & \vdots & \ddots & \vdots & \vdots \\
                0 & 0 & \cdots & I & 0 & \cdots & 0 & 0 \\
                0 & 0 & \cdots & 0 & 0 & \cdots & 0 & 0 \\
                \hline
                0 & 0 & \cdots & 0 & 0 & I & \cdots & 0 \\
                \vdots & \vdots & \ddots & \vdots & \vdots & \ddots & \vdots & \vdots \\
                0 & 0 & \cdots & 0 & 0 & 0 & \cdots & I \\
                CA^{L-1} B_0 & CA^{L-2} B_0 & \cdots & CB_0 & 0 & 0 & \cdots & 0
            \end{array}\right]
            + \left[\begin{array}{c}
                0 \\
                \vdots \\
                0 \\
                0 \\
                \hline
                0 \\
                \vdots \\
                0 \\
                CA^L (\hat{\cH}_L^{(z_0)})^\dagger \begin{bmatrix} -\hat{\cH}_L^{(u)} & I_{pL} \end{bmatrix}
            \end{array}\right],
            \\
            \label{eq:extended-state-matrices-linear-B-C-D}
            \tB &= \left[\begin{array}{cccc|cccc}
                0 &
                \cdots &
                0 &
                I_m &
                0 &
                \cdots &
                0 &
                D
            \end{array}\right]^\top,
            \quad 
            \tC = \left[\begin{array}{cccc|cccc}
                0 & 0 & \cdots & 0 & 0 & 0 & \cdots & I
            \end{array}\right]\tA,
            \quad 
            \tD = D.
        \end{align}%
    \end{subequations}%
\end{proposition}
\begin{proof}
    This result is well-known; see, e.g.,~\citet{koch:berberich:allgower:2022,ho:kalman:1966}.
    For completeness, we will include a short proof in the following.

    First, we rewrite~\eqref{eq:IO-representation-linear-yseq} as $
        \begin{bmatrix} -\hat{\cH}_L^{(u)} & I_{pL} \end{bmatrix}
        \xi_k
        = \hat{\cH}_L^{(z_0)}  z_{k-L}
    $.
    Due to the definition of the lag $\nu$, $\hat{\cH}_L^{(z_0)}$ has full column rank for any length $L\geq \nu$ and, thus, the left-inverse
    \begin{equation}
        (\hat{\cH}_L^{(z_0)})^\dagger 
        = (\hat{\cH}_L^{(z_0)})^\top \hat{\cH}_L^{(z_0)})^{-1} (\hat{\cH}_L^{(z_0)})^\top.
    \end{equation}%
    exists.
    Hence, $z_{k-L} = (\hat{\cH}_L^{(z_0)})^\dagger \begin{bmatrix} -\hat{\cH}_L^{(u)} & I_{pL} \end{bmatrix} \xi_k$.
    Further, by recursively evaluating~\eqref{eq:dynamics-linear}, we observe
    \begin{align}
        y_k 
        &= CA^L z_{k-L} + \begin{bmatrix}
            CA^{L-1} B_0 & CA^{L-2} B_0 & \cdots & CB_0
        \end{bmatrix}
        \begin{bmatrix}
            u_{k-L} \\ u_{k-L+1} \\ \vdots \\ u_{k-1}
        \end{bmatrix}
        + D u_k
        \\
        &= CA^L (\hat{\cH}_L^{(z_0)})^\dagger \begin{bmatrix} -\hat{\cH}_L^{(u)} & I_{pL} \end{bmatrix} \xi_k + \begin{bmatrix}
            CA^{L-1} B_0 & CA^{L-2} B_0 & \cdots & CB_0
        \end{bmatrix}
        \begin{bmatrix}
            u_{k-L} \\ u_{k-L+1} \\ \vdots \\ u_{k-1}
        \end{bmatrix}
        + D u_k.
    \end{align}%
    Thus, by additionally exploiting the shift-structure of past inputs and outputs in $\xi_k$, we have established~\eqref{eq:IO-representation-linear} with the matrices $\tA$, $\tB$, $\tC$, and $\tD$ in~\eqref{eq:extended-state-matrices-linear}.
\end{proof}
Both representations~\eqref{eq:dynamics-linear} and~\eqref{eq:IO-representation-linear} share the same input-output behavior.
Thus, the extended state allows for a constructive system representation from input-output data.

\end{document}